\documentclass[11 pt, journal, draftcls, onecolumn]{IEEEtran}
\IEEEoverridecommandlockouts

\usepackage{amsmath, mathrsfs, graphicx,epsfig,amssymb, amsfonts, color, subfigure}
\usepackage[center]{caption}

\newtheorem{lemma}{\bf Lemma}

\newtheorem{defn}{\bf Definition}
\newtheorem{thm}{\bf Theorem}

\newtheorem{proof}{\bf Proof}

\newtheorem{rem}{\bf Remark}
\newtheorem{cor}{\bf Corollary}
\newtheorem{ex}{\bf Example}

\begin{document}

\title{\huge{The Generalized Degrees of Freedom Region of the MIMO Interference Channel}}

\author{{Sanjay Karmakar ~~~~ Mahesh K.~Varanasi}
\thanks{This work was supported in part by NSF Grant CCF-0728955.
The authors are with the Department of Electrical, Computer, and
Energy Engineering, University of Colorado, Boulder, CO 80309-0425
USA (e-mail: {sanjay.karmakar, varanasi}@colorado.edu).} }

\maketitle

\begin{abstract}
The generalized degrees of freedom (GDoF) region of the MIMO Gaussian interference channel (IC) is obtained for the general case of an arbitrary number of antennas at each node and where the signal-to-noise ratios (SNR) and interference-to-noise ratios (INR) vary with arbitrary exponents to a nominal SNR. The GDoF region reveals various insights through the joint dependence of optimal interference management techniques (at high SNR) on the SNR exponents that determine the relative strengths of direct-link SNRs and cross-link INRs and the numbers of antennas at the four terminals. For instance, it permits an in-depth look at the issue of rate-splitting and partial decoding and it reveals that, unlike in the scalar IC, treating interference as noise is not always GDoF-optimal even in the very weak interference regime. Moreover, while the DoF-optimal strategy that relies just on transmit/receive zero-forcing beamforming and time-sharing is not GDoF optimal (and thus has an unbounded gap to capacity), the precise characterization of the very strong interference regime -- where single-user DoF performance can be achieved simultaneously for both users-- depends on the relative numbers of antennas at the four terminals and thus deviates from what it is in the SISO case. For asymmetric numbers of antennas at the four nodes the shape of the symmetric GDoF curve can be a ``distorted W" curve to the extent that for certain MIMO ICs it is a ``V" curve.
\end{abstract}

\IEEEpeerreviewmaketitle

\newpage

\section{Introduction}

In spite of research spanning over three decades, the capacity of the interference channel (IC) has been characterized only for some special cases (cf. \cite{Carleial75,Benzel,Sato,Gamal_Costa,Motahari_Khandani,Shang2009,Annapureddy_Veeravalli,SCKP}) and remains an open problem to date for the general case. However, the understanding of the Gaussian interference channel has been significantly enriched through those works and through capacity approximations in recent years. In particular, the capacity region was characterized to within one bit for the single-antenna (SISO) IC \cite{ETW1} and to within a constant number of bits for the more general case of MIMO Gaussian ICs with an arbitrary number of antennas at each node in \cite{TT,Sanjay_Varanasi_Cap_MIMO_IC_const_gap}. Moreover, the degrees of freedom (DoF) region of the general MIMO IC was obtained in \cite{JFak}.

The constant gap capacity approximations of \cite{ETW1} and \cite{Sanjay_Varanasi_Cap_MIMO_IC_const_gap} provide performance guarantees to within a constant number of bits on SISO and MIMO ICs, respectively, by showing that simple Han-Kobayashi (HK) \cite{Han_Kobayashi} coding schemes with Gaussian inputs and no time-sharing can achieve the capacity region to within a constant number of bits, with the constant being independent of the SNRs, INRs and the channel matrices. The capacity approximation through the DoF characterization of \cite{JFak} shows that simple transmit beamforming and receive zero-forcing is sufficient to preserve the DoF optimality of the MIMO IC.  While in the DoF characterization, it is assumed that the SNRs and interference-to-noise ratios (INRs) at each receiver scale in a similar way, {\em asymmetric} scaling of the SNRs and INRs on the {\em dB scale} has been shown to provide more insight about optimal interference management at high SNR as a function of that asymmetry. This phenomenon was captured succinctly through the so-called {\em generalized degrees of freedom} metric in the context of the SISO IC in \cite{ETW1}. This work obtains the GDoF region for the general MIMO IC, thereby generalizing the same result for the SISO IC obtained in \cite{ETW1} and the usual DoF region for general MIMO IC obtained in \cite{JFak}.

The GDoF region metric, as its name suggests, generalizes the notion of the conventional degrees of freedom (DoF) region metric by additionally emphasizing the signal level as a signaling dimension. It therefore characterizes the simultaneously accessible fractions of spatial and signal-level dimensions (per channel use) by the two users in the limit of high signal-to-noise ratio (SNR) while the ratios of the SNRs and INRs relative to a reference SNR, each expressed in the dB scale, are held constant, with each constant taken, in the most general case, to be arbitrary. The GDoF region was obtained for the SISO IC in \cite{ETW1} based on the constant gap to capacity result found therein. The symmetric GDoF, $d_{\rm sym} (\alpha) $, which is the maximum common GDoF achievable by each of the two users, for the symmetric SISO IC with equal SNRs and equal INRs for the two users, i.e, with ${\rm INR} = {\rm SNR}^\alpha $ was evaluated in \cite{ETW1} to be the well-known ``W" curve. The W-curve clearly delineates the very weak, weak, moderate, strong and very strong interference regimes, depending on the value of $\alpha$, pointing to the optimal (upto GDoF accuracy) interference management techniques as a function of the severity or mildness of the interference.

There have been several other recent works on characterizing the GDoF of various channels. For example, in \cite{PBT}, the symmetric GDoF of a class of symmetric MIMO ICs -- for which the SNRs at each receiver are the same and the INRs at each receiver are also the same, with ${\rm INR} = {\rm SNR}^\alpha $-- and where both transmitters have $M$ antennas and both receivers have $N$ antennas, with
the restriction $N \geq M$, was obtained and found to be a ``W" curve also. In \cite{Jafar_Vishwanath_GDOF}, the symmetric GDoF in the perfectly symmetric (with all direct links having identical gains and all cross links having identical gains) scalar $K$-user interference network was found (see also \cite{Bandemer-Cyclic}). In \cite{Gou_Jafar}, the symmetric GDoF was obtained for the $(N+1)$-user symmetric SIMO IC with $N$ antennas at each receiver and with equal direct link SNRs and equal cross link INRs. The symmetric GDoF of a symmetric model of the scalar X-channel with real-valued channel coefficients was found in \cite{Huang_Cadambe_Jafar}.

In this work, we obtain the GDoF region of the general MIMO IC with an arbitrary number of antennas at each node and in the most general case where the signal-to-noise ratios (SNR) and interference-to-noise ratios (INRs) vary with arbitrary exponents to a nominal SNR. This result is made possible by the recent constant gap to capacity characterization for the general MIMO IC in \cite{TT,Sanjay_Varanasi_Cap_MIMO_IC_const_gap}. The GDoF result of this paper thus generalizes the GDoF region of the SISO IC found in \cite{ETW1} to the MIMO IC. It also recovers the symmetric GDoF result of \cite{PBT} for the class of symmetric MIMO ICs considered therein. Moreover, the single and unified constant-gap-to-capacity achievability scheme of \cite{Sanjay_Varanasi_Cap_MIMO_IC_const_gap} considered here, unlike that in \cite{PBT}, does not require the restriction on the numbers of antennas at the different nodes or on the values of SNR exponents and is GDoF optimal in the most general case. The main result of this work also recovers the conventional DoF region result obtained in \cite{JFak} for the MIMO IC by setting all SNR exponents to unity. In addition to providing several insights that include whatever is common between certain symmetric (in numbers of antennas) MIMO ICs and SISO ICs and what is not, the GDoF result of this paper gives rise to new insights into optimal signaling strategies that make jointly optimal use of the available spatial and signal level dimensions.

The single, unified achievable scheme studied in depth here that is GDoF optimal (and indeed constant-gap-to-capacity optimal \cite{Sanjay_Varanasi_Cap_MIMO_IC_const_gap}) is a simple Han-Kobayashi coding scheme with mutually independent Gaussian input for the private and public messages of each user without time-sharing. The private and public message can be thought of consisting of several information streams. The private information streams are either directed along the null space of the corresponding cross-link channel matrix or transmitted at power levels that ensure that they reach the unintended receiver below the noise floor. Such a scheme therefore jointly and optimally employs both {\em signal-level interference alignment} \cite{Bresler-parekh-tse} as well as {\em transmit beamforming or signal-space interference alignment} \cite{Cadambe-Jafar-IA} techniques. For a given DoF tuple in the GDoF region of the channel, we also explicitly specify the DoFs carried by the private and public messages of each user.

The rest of the paper is organized as follows. In Section~\ref{sec_channel_model_and_preliminaries} we describe the channel model and the GDoF optimal coding scheme. Section \ref{sec_main_result} contains the main result of this paper, namely, the GDoF region of the general MIMO IC. Specializations of this result to the SISO IC and to the DoF region of the MIMO IC are also given in Section \ref{sec_main_result} which recover the results of \cite{ETW1} and \cite{JFak}, respectively. Explicit specifications of the DoF-splitting between private and public sub-messages are obtained. The reciprocity property of the GDoF region (which denotes the invariability of GDoF with respect to direction of information flow) is described in \ref{sec:reciprocity-of-GDoF-region} as are specializations of the GDoF region to obtain the symmetric GDoF of the symmetric $(M,N,M,N) $ MIMO IC, thereby recovering as a special case the result of \cite{PBT} for $M \leq N$. Section~\ref{sec:explicit-examples} gives in-depth descriptions,  through 2 specific examples of weak and mixed interference MIMO ICs, of how different DoF tuples in the GDoF region are achieved through the joint specification of the strategies at the transmitters and the receivers. In Section~\ref{sec_insights} several novel insights revealed by the GDoF analysis are given. Section~\ref{sec_conclusion} concludes the paper.

\begin{proof}[Notations]
Let $\mathbb{C}$ and $\mathbb{R}^+$ represent the field of complex numbers and the set of non-negative real numbers, respectively. An $n\times m$ matrix with entries coming from $\mathbb{C}$ will be denoted by $A\in \mathbb{C}^{n\times m}$ and its entry in the $i^{th}$ row and $j^{th}$ column will be denoted by $[A]_{ij}$. We shall denote the transpose and the conjugate transpose of the matrix $A$ by $A^T$ and $A^{\dagger}$ respectively. $I_n$ represents the $n\times n$ identity matrix, $0_{m\times n}$ represents an all zero $m\times n$ matrix and $\mathbb{U}^{n\times n}$ represents the set of $n\times n$ unitary matrices. The $k^{th}$ column (row) of the matrix $A$ will be denoted by $A^{[k]}$ ($A^{(k)}$) whereas $A^{[k_1:k_2]}$ ($A^{(k_1:k_2)}$) will represent a matrix whose columns (rows) are same as the $k_1^{th}$ to $k_2^{th}$ columns (rows) of matrix $A$. If $x^{(k)}\in \mathbb{C}, \forall~ 1\leq k\leq n$, then $\mathbf{x}\triangleq [x^{(1)}, x^{(2)},\cdots , x^{(n)}]^T$. $\{A,B,C,D\}$ will represent an ordered set of matrices. $I(x;y)$ and $I(x;y|z)$ will represent the mutual information and conditional mutual information of the arguments, respectively. $(x\land y)$, $(x\vee y)$ and $(x)^+$ represents the minimum and maximum between $x$ and $y$ and maximum between $x$ and $0$, respectively. We also use Landau notations for error terms in approximations. $o(1)$ denotes a term which goes to zero asymptotically and $\mathcal{O}(1)$ denotes a term which is bounded above by some constant. We say $x$ is of the order of $y$ if $\lim_{y\to \infty} \frac{x}{y}=0$. All the logarithms in this paper are with base $2$. We denote the distribution of a complex circularly symmetric Gaussian random vector with zero mean and covariance matrix $Q$, by $\mathcal{CN}(0,Q)$. Finally, the indicator function $1(\textrm{S})$ is defined as follows
\begin{equation*}
1(\textrm{S})=\left\{\begin{array}{c}
1, ~~\textrm{if S is true};\\
0, ~~\textrm{if S is false}.
\end{array}\right.
\end{equation*}
\end{proof}

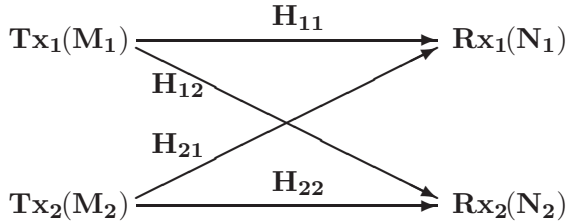
\begin{figure}[!thb]
\setlength{\unitlength}{1mm}
\begin{picture}(80,40)
\thicklines
\put(60,10){\vector(1,0){40}}
\put(60,32){\vector(1,0){40}}
\put(60,11){\vector(2,1){40}}
\put(60,31){\vector(2,-1){40}}
\put(50,9){$\mathbf{(M_2)}$}
\put(109,9){$\mathbf{(N_2)}$}
\put(50,31){$\mathbf{(M_1)}$}
\put(109,31){$\mathbf{(N_1)}$}

\put(43,9){$\mathbf{Tx_2}$}
\put(102,9){$\mathbf{Rx_2}$}
\put(43,31){$\mathbf{Tx_1}$}
\put(102,31){$\mathbf{Rx_1}$}

\put(78,12){$\mathbf{H_{22}}$}
\put(78,34){$\mathbf{H_{11}}$}
\put(62,17){$\mathbf{H_{21}}$}
\put(62,25){$\mathbf{H_{12}}$}
\end{picture}
\caption{The $(M_1,N_1,M_2,N_2)$ MIMO IC.}
\label{channel_model_two_user_IC}
\end{figure}

\section{Channel Model and preliminaries }
\label{sec_channel_model_and_preliminaries}

In this section, we define the two-user MIMO IC, its capacity region and the generalized degrees of freedom region and state the upper and lower bounds on the capacity region that are within a constant gap of each other as obtained in \cite{Sanjay_Varanasi_Cap_MIMO_IC_const_gap} for the sake of completeness. We describe the achievable scheme of \cite{Sanjay_Varanasi_Cap_MIMO_IC_const_gap} and then give asymptotic (high SNR) approximations upto $O(1)$ of key quantities that arise in the bounds on capacity region. These approximations are used later to derive the main result of this paper on the GDoF region of the MIMO IC.

\subsection{The MIMO IC}

The 2-user MIMO IC, with $M_i$ and $N_i$ antennas at transmitter $i$ ($Tx_i$) and receiver $i$ ($Rx_i$), respectively, for $i=1,~2$ as shown in Fig. \ref{channel_model_two_user_IC} (hereafter referred to as $(M_1,N_1,M_2,N_2)$ IC) is considered. We consider a time-invariant or fixed channel where the channel matrices, $H_{ij}$'s remain fixed for the entire duration of communication. We also assume that the entries of these matrices are drawn i.i.d. from a continuous and unitarily invariant~\cite{Tulino_Verdu} distribution, i.e., $UH_{ij}V$ is identically distributed to $H_{ij}$ for any $U\in \mathbb{U}^{N_j\times N_j}$ and $V\in \mathbb{U}^{M_i\times M_i}$ which ensures that the channel matrices are full rank with probability one (w.p.1). This class of distributions will be denoted by $\mathcal{P}$ in the rest of the paper. In the channel model, we also incorporate a real-valued path-loss or attenuation factor, denoted as $\eta_{ij}$, for the signal transmitted from $Tx_i$ to receiver $Rx_j$. At time $t$, $Tx_i$ chooses a vector ${X}_{it}\in \mathbb{C}^{M_i\times 1}$ and sends $\sqrt{P_i}{X}_{it}$ over the channel, where we assume the following average input power constraint at $Tx_i$,
\begin{equation}
\label{power_constraint}
\frac{1}{n}\sum_{t=1}^{n}\textrm{Tr} (Q_{it}) ~\leq ~1,
\end{equation}
for $ i \in \{ 1,2 \}$, where $ Q_{it}=\mathbb{E}({X}_{it}{X}_{it}^{\dagger})$ and $Q_{it}$'s can depend on the channel matrices. The received signals at time $t$ can be written as
\begin{eqnarray}
\label{system_eq_two_user_IC2}
Y_{1t}=\sqrt{\rho_{11}}H_{11}X_{1t}+ \sqrt{\rho_{21}}H_{21}X_{2t}+Z_{1t},\\
Y_{2t}= \sqrt{\rho_{22}}H_{22}{X}_{2t}+\sqrt{\rho_{12}} H_{12}X_{1t}+Z_{2t},
\end{eqnarray}
where $Z_{it}\in\mathbb{C}^{N_i\times 1}$ are i.i.d $\mathcal{CN}(\mathbf{0}, I_{N_i})$ across $i$ and $t$, $\rho_{ii}=\eta_{ii}\sqrt{P_i}$ represents the signal-to-noise ratio (SNR) at receiver $i$ and $\rho_{ij}=\eta_{ij}\sqrt{P_i}$ represents the interference-to-noise ratio (INR) at receiver $j$ for $i \neq j\in \{1,2\}$.

The performance on the MIMO IC should depend on the strength of the interference relative to the desired signal level on the dB scale with, for example, a better DoF performance expected when interference strength is much less or much higher than the signal strength as in the SISO IC \cite{ETW1}. This variation of performance due to relative difference in strengths of SNRs and INRs can not be captured (at high SNR) by a DoF analysis alone, i.e., if they differ say by only a constant. To characterize the DoF region under such a scenario we thus let the SNRs and INRs vary exponentially with respect to a nominal SNR, $\rho$, with different scaling factors as follows:
\begin{IEEEeqnarray}{l}
\label{eq_scaling_of_snrinr}
\lim_{\log(\rho)\to \infty}\frac{\log (\rho_{ij})}{\log(\rho)}=\alpha_{ij}, ~\textrm{where}~\alpha_{ij}\in \mathbb{R}^+~\textrm{and}~ i,j\in \{1,2\}.
\end{IEEEeqnarray}
Without loss of generality, we assume $\rho_{11}=\rho$ or $\alpha_{11}=1$ throughout the rest of the paper. As mentioned earlier, this technique of varying different SNRs and INRs was first introduced in~\cite{ETW1} to characterize the DoF region of the SISO $2$-user IC and the corresponding DoF region was called the {\em generalized} DoF (GDoF) region.


\subsection{Capacity region bounds to within a constant gap}

In what follows, the MIMO IC with the channel matrices, SNRs and INRs as described above will be denoted by $\mathcal{IC}\left(\mathcal{H},\bar{\rho}\right)$, where $\mathcal{H}=\{H_{11},H_{12},H_{21},H_{22}\}$ and $\bar{\rho}=[\rho_{11},\rho_{12},\rho_{21},\rho_{22}]$ or equivalently as $\mathcal{IC}\left(\mathcal{H},\bar{\alpha}\right)$ where $\bar{\alpha}=[\alpha_{11},\alpha_{12},\alpha_{21},\alpha_{22}]$. In what follows, $\bar{\rho}$ and $\bar{\alpha}$ will be used interchangeably to indicate the power levels of different links of the channel. The capacity region of $\mathcal{IC}\left(\mathcal{H},\bar{\alpha}\right)$ is defined in the usual way (cf. \cite{Sanjay_Varanasi_Cap_MIMO_IC_const_gap}) and will be denoted by $\mathcal{C}\left(\mathcal{H},\bar{\alpha}\right)$. Inner and outer bounds from \cite{Sanjay_Varanasi_Cap_MIMO_IC_const_gap} are stated next.

\subsubsection{An outer bound to the capacity region}
\label{sub_upper_bound}
\begin{lemma}[Lemma~1 of \cite{Sanjay_Varanasi_Cap_MIMO_IC_const_gap}]
\label{lem_upper_bound}
For a given $\mathcal{H}$ and $\bar{\alpha}$ the capacity region, $\mathcal{C}(\mathcal{H},\bar{\alpha})$ of a $2$-user MIMO Gaussian IC, with input power constraint \eqref{power_constraint}, is contained within the set of rate tuples $\mathcal{R}^u(\mathcal{H},\bar{\alpha})$, i.e.,
\[\mathcal{C}(\mathcal{H},\bar{\alpha})~\subseteq~\mathcal{R}^u(\mathcal{H},\bar{\alpha}),\]
where $\mathcal{R}^u(\mathcal{H},\bar{\alpha})$ represents the set of rate pairs $(R_1,R_2)$, satisfying the following constraints:
\begin{IEEEeqnarray}{rl}
\label{eq_bound1}
R_1\leq \log \det &\left(I_{N_1}+\rho_{11} H_{11}H_{11}^{\dagger}\right);\\
\label{eq_bound2}
R_2\leq \log \det &\left(I_{N_2}+\rho_{22} H_{22}H_{22}^{\dagger}\right); \\
\label{eq_bound3}
R_1+R_2\leq \log \det &\left(I_{N_2}+\rho_{12} H_{12}H_{12}^{\dagger}+\rho_{22} H_{22}H_{22}^{\dagger}\right)
 +\log \det \left(I_{N_1}+\rho_{11} H_{11}K_1H_{11}^{\dagger}\right);\\
\label{eq_bound4}
R_1+R_2\leq \log \det &\left(I_{N_1}+\rho_{21} H_{21}H_{21}^{\dagger}+\rho_{11} H_{11}H_{11}^{\dagger}\right)
 +\log \det \left(I_{N_2}+\rho_{22} H_{22}K_2H_{22}^{\dagger}\right);\\
R_1+R_2\leq \log \det &\left(I_{N_1}+ \rho_{21} H_{21}H_{21}^{\dagger}+  \rho_{11} H_{11} K_1H_{11}^{\dagger}\right)\nonumber \\
\label{eq_bound5}
 & +\log \det \left(I_{N_2}+\rho_{12} H_{12}H_{12}^{\dagger}+\rho_{22} H_{22} K_2H_{22}^{\dagger}\right);\\
2R_1+R_2\leq \log \det &\left(I_{N_1}+\rho_{21}  H_{21}H_{21}^{\dagger}+ \rho_{11} H_{11}H_{11}^{\dagger}\right)+\log\det\left(I_{N_1}+\rho_{11} H_{11} K_1H_{11}^{\dagger}\right)+ \nonumber \\
\label{eq_bound6}
& \log \det \left(I_{N_2}+ \rho_{12}  H_{12}H_{12}^{\dagger}+ \rho_{22} H_{22} K_2H_{22}^{\dagger}\right);
\end{IEEEeqnarray}
\begin{IEEEeqnarray}{rl}
R_1+2R_2\leq \log \det &\left(I_{N_2}+\rho_{12}  H_{12}H_{12}^{\dagger}+ \rho_{22} H_{22}H_{22}^{\dagger}\right)+\log\det\left(I_{N_2}+\rho_{22} H_{22} K_2H_{22}^{\dagger}\right)+ \nonumber \\
\label{eq_bound7}
&\log \det \left(I_{N_1}+ \rho_{21} H_{21}H_{21}^{\dagger}+\rho_{11} H_{11} K_1H_{11}^{\dagger}\right),
\end{IEEEeqnarray}
where $K_i=M_iK_{iu}$ and $K_{iu}$ is as specified in equation \eqref{eq_power_split} and $\rho_{ij}=\rho^{\alpha_{ij}}$ for $i, j\in\{1,2\}$.
\end{lemma}

\subsubsection{An inner bound to the capacity region}
\label{subsection_HK_I}
A simple HK coding scheme, denoted as $\mathcal{HK}\left(K_{1u},K_{1w},K_{2u},K_{2w}\right)$ is described in detail in the next subsection. Here we state an inner bound to the rate region achievable using this scheme found in \cite{Sanjay_Varanasi_Cap_MIMO_IC_const_gap} which is hence also an inner bound to the capacity region.
\begin{lemma}
\label{lem_achievable_region}
Let us denote the achievable rate region of $\mathcal{HK}\left(K_{1u},K_{1w},K_{2u},K_{2w}\right)$ on $\mathcal{IC}(\mathcal{H},\bar{\alpha})$ by $\mathcal{R}_a(\mathcal{H},\bar{\alpha})$, then $\mathcal{R}_a(\mathcal{H},\bar{\alpha})$ is the set of rate tuples $(R_1,R_2)\in \mathbb{R}^{+2}$ satisfying the following conditions:
\begin{IEEEeqnarray*}{rl}
R_1\leq \log \det &\left(I_{N_1}+\rho_{11} H_{11}H_{11}^{\dagger}\right)-n_1;\\
R_2\leq \log \det &\left(I_{N_2}+\rho_{22} H_{22}H_{22}^{\dagger}\right)-n_2; \\
R_1+R_2\leq \log \det &\left(I_{N_2}+\rho_{12} H_{12}H_{12}^{\dagger}+\rho_{22} H_{22}H_{22}^{\dagger}\right)\\
&+\log \det \left(I_{N_1}+\rho_{11} H_{11}K_{1}H_{11}^{\dagger}\right)-(n_1+n_2);\\
R_1+R_2\leq \log \det &\left(I_{N_1}+\rho_{21} H_{21}H_{21}^{\dagger}+\rho_{11} H_{11}H_{11}^{\dagger}\right)\\
 &+\log \det \left(I_{N_2}+\rho_{22} H_{22}K_{2}H_{22}^{\dagger}\right)-(n_1+n_2);\\
R_1+R_2\leq \log \det &\left(I_{N_1}+ \rho_{21} H_{21}H_{21}^{\dagger}+ \rho_{11} H_{11}K_{1}H_{11}^{\dagger}\right)\\
& +\log \det \left(I_{N_2}+\rho_{12} H_{12}H_{12}^{\dagger}+\rho_{22} H_{22} K_{2} H_{22}^{\dagger}\right) -(n_1+n_2);\\
2R_1+R_2\leq \log \det &\left(I_{N_1}+\rho_{21}  H_{21}H_{21}^{\dagger}+ \rho_{11} H_{11}H_{11}^{\dagger}\right)+\log\det\left(I_{N_1}+\rho_{11} H_{11} K_{1}H_{11}^{\dagger}\right)+ \\
& \log \det \left(I_{N_2}+ \rho_{12}  H_{12}H_{12}^{\dagger}+\rho_{22} H_{22} K_{2}H_{22}^{\dagger}\right)-(2n_1+n_2);\\
R_1+2R_2\leq \log \det &\left(I_{N_2}+\rho_{12}  H_{12}H_{12}^{\dagger}+ \rho_{22} H_{22}H_{22}^{\dagger}\right)+\log\det\left(I_{N_2}+ \rho_{22} H_{22} K_{2}H_{22}^{\dagger}\right)+ \\
&\log \det \left(I_{N_1}+ \rho_{21} H_{21}H_{21}^{\dagger}+\rho_{11} H_{11} K_1 H_{11}^{\dagger}\right)-(n_1+2n_2),
\end{IEEEeqnarray*}
where $n_i$ for $i=1,2$ are constants functions of the number of antennas only, $K_{iw}$ and $K_{iu}$ are given by equation \eqref{eq_power_split} and $\rho_{ij}=\rho^{\alpha_{ij}}$ for $i, j\in\{1,2\}$..
\end{lemma}

\subsection{The simple HK coding scheme}
\label{def_coding_scheme}
We describe next the coding scheme whose rate region contains the rate region of of Lemma \ref{lem_achievable_region} which is in turn within a constant gap to the capacity region.
Let each of the users divide its message into two sub-messages (called the private and public messages hereafter) and use superposition coding to
encode the two sub-messages with mutually independent random zero-mean Gaussian codewords so that we have
\begin{equation}
\begin{array}{c}
X_1={U}_1+{W}_1,\\
X_2={U}_2+{W}_2,
\end{array}
\end{equation}
where ${U}_i$ and ${W}_i$ represent the codewords of the private and public messages of user $i$, respectively\footnote{With a slight abuse of notation, we shall use the same symbol to represent a message and its corresponding codeword.}. Moreover, the covariance matrices of the public and private messages are taken for each $i \in \{1,2\}$ to be
\begin{IEEEeqnarray}{rl}
\label{eq_power_split}
K_{iu}(\mathcal{H})\triangleq\mathbb{E}(U_i U_i^{\dagger}) =& \frac{1 }{M_{i}}\left(I_{M_i}+\rho_{ij} H_{ij}^{\dagger}H_{ij}\right)^{-1},\nonumber \\
K_{iw}(\mathcal{H})\triangleq\mathbb{E}(W_i W_i^{\dagger})=&  \left(\frac{I_{M_i}}{M_{i}}-K_{iu}\right).
\end{IEEEeqnarray}
In what follows, we shall refer to such a coding scheme as the $\mathcal{HK}\left(\{K_{1u},K_{1w},K_{2u},K_{2w}\}\right)$ scheme, where we drop the dependence of $K_{iu}$ and $K_{iw}$ on the channel matrices for notational convenience. Let the singular value decomposition (SVD) of the channel matrix $H_{ij}$ be given by $H_{ij}=V_{ij}\Sigma_{ij}U_{ij}^\dagger$, where $V_{ij}\in \mathbb{U}^{N_j\times N_j}$ and $U_{ij}\in \mathbb{U}^{M_i\times M_i}$ are unitary matrices and $\Sigma_{ij}\in \mathbb{C}^{N_j\times M_i}$ is a rectangular matrix containing the singular values along its diagonal. Using the SVD of the matrix $H_{ij}$, the covariance matrices for $U_i$ and $W_i$ of equation \eqref{eq_power_split} can alternatively be written as
\begin{IEEEeqnarray}{rl}
\label{eq:expansion-of-kiu}
K_{iu}=U_{ij}\left[\begin{array}{cc} \frac{1}{M_i}\left(I_{\min\{M_i,N_j\}}+\rho^{\alpha_{ij}}\Lambda_{ij}\right)^{-1} & \mathbf{0}\\ \mathbf{0}& \frac{1}{M_i}I_{(M_i-N_j)^+}\end{array}\right] U_{ij}^\dagger =U_{ij}D_{ij} U_{ij}^\dagger,
\end{IEEEeqnarray}
where $\Lambda_{ij}$ is a diagonal matrix containing the non-zero eigenvalues of $H_{ij}^\dagger H_{ij}$ and denoting the quantity $M_i(1+\rho^{\alpha_{ij}}\lambda_{ij}^{(k)})=r_{ik}$ where $\lambda_{ij}^{(k)}$ is the $k^{th}$ non-zero eigenvalues of $H_{ij}^\dagger H_{ij}$ for $1\leq k\leq m_{ij}=\min\{M_i,N_j\}$ we have
\begin{IEEEeqnarray}{rl}
[D_{ij}]_{kk}=\left\{\begin{array}{cc}r_{ik}^{-1},& \textrm{for} ~1\leq k\leq m_{ij};\\
\frac{1}{M_i},& \textrm{for}~m_{ij}+1\leq k\leq M_i.
\end{array}\right.
\end{IEEEeqnarray}
Similarly, we have
\begin{IEEEeqnarray}{rl}
\label{eq:expansion-of-kiw}
K_{iw}=U_{ij}\left(\frac{1}{M_i}I_{M_i}-D_{ij}\right) U_{ij}^\dagger = U_{ij}^{[1:m_{ij}]}\widetilde{D}_{ij}(U_{ij}^{[1:m_{ij}]})^\dagger,
\end{IEEEeqnarray}
where $[\widetilde{D}_{ij}]_{kk}=(\frac{1}{M_i}-\frac{1}{r_{ik}})$ for $1\leq k\leq m_{ij}$. Now, it is well known that a Gaussian vector $V$ with covariance matrix $K$ can be expressed as $V=Ax$, where $x$ is a Gaussian vector with identity as covariance matrix if $AA^\dagger=K$. Using this result along with equations \eqref{eq:expansion-of-kiu} and \eqref{eq:expansion-of-kiw} we can write
\begin{IEEEeqnarray}{rl}
U_i=&U_{ij} \sqrt{D_{ij}}\mathbf{x}_{ip}=\sum_{l=1}^{M_i}\sqrt{[D_{ij}]_{ll}}x_{ip}^{(l)}U_{ij}^{[l]};\nonumber \\
W_i=&U_{ij} \sqrt{\widetilde{D}_{ij}}\mathbf{x}_{ic}= \sum_{k=1}^{m_{ij}}\sqrt{[\widetilde{D}_{ij}]_{kk}}x_{ic}^{(k)}U_{ij}^{[k]},
\end{IEEEeqnarray}
where $\mathbf{x}_{ic}=[x_{ic}^{(1)}, \cdots , x_{ic}^{(m_{ij})}]^T\sim \mathcal{CN}(0,I_{m_{ij}})$ and $\mathbf{x}_{ip}=[x_{ip}^{(1)}, \cdots , x_{ip}^{(M_i)}]^T\sim \mathcal{CN}(0,I_{M_i})$ are mutually independent normal Gaussian vectors. Substituting this in the expression for $X_i$ we see that, the transmit signal at $Tx_i$ can be written as
\begin{IEEEeqnarray}{rl}
\label{eq_structure_of_streams}
X_i=\sum_{k=1}^{m_{ij}}\sqrt{[\widetilde{D}_{ij}]_{kk}}&x_{ic}^{(k)}U_{ij}^{[k]}+\sum_{l=1}^{m_{ij}}\sqrt{[D_{ij}]_{ll}}x_{ip}^{(l)} U_{ij}^{[l]}+ \sum_{m=1+m_{ij}}^{M_i}\frac{1}{\sqrt{M_i }}x_{ip}^{(m)}U_{ij}^{[m]}.
\end{IEEEeqnarray}
In the above equation, $x_{ic}^{(l)}$ for $1\leq l\leq m_{ij}$ and $x_{ip}^{(k)}$ for $1\leq k\leq M_i$ represent the $l^{th}$ and $k^{th}$ stream of the public and private information along directions $U_{ij}^{[l]}$ and $U_{ij}^{[k]}$, respectively, for user $i$.

\begin{rem}
\label{rem_use_of_null_space}
Note that each of the terms in the second sum of the right hand side of (\ref{eq_structure_of_streams}) have power proportional to $\rho_{ij}^{-1}$. Hence, all the streams encoded through $x_{ip}^{(k)}$ for $1\leq k\leq m_{ij}$ after passing through the cross channel with strength $\rho_{ij}$ reach $Rx_j$ at the noise floor. This technique can be considered as a form of {\em interference alignment at the signal level}.

On the other hand, in the SVD of the matrix $H_{ij}$, the last $(M_i-N_j)^+$ columns of $\Sigma_{ij}$ are all zeros and hence $H_{ij}U_{ij}^{[k]}=0$ for $N_j<k\leq M_i$. In other words, each of the $U_{ij}^{[k]}$'s for $m_{ij}<k\leq M_{i}$ lie in the null space of the matrix $H_{ij}$. Therefore, any stream sent along one of these directions reaches $Rx_j$ in a subspace which is perpendicular to the subspace in which the useful signals of $Rx_2$ lie. That is, each user can be said to align the interference to the undesired user in a particular subspace, which is a simple form of {\em signal space interference alignment}. This explains why we call the streams carried by $x_{ip}^{(k)},~1\leq k\leq M_{i}$, {\it private} streams.

Thus the specific choice of the covariance matrices $K_{iu}$ in $\mathcal{HK}\left(\{K_{1u},K_{1w},K_{2u},K_{2w}\}\right)$ for $i=1,2$ amounts to employing a technique to jointly utilize both types of interference alignments described above.
\end{rem}

\subsection{Generalized Degrees of Freedom Region}

\begin{defn}
\label{def:GDoF-region}
The GDoF region, $\mathcal{D}_o(\bar{M},\bar{\alpha})$, of $\mathcal{IC}(\mathcal{H},\bar{\alpha})$ is defined as
\begin{IEEEeqnarray}{rl}
\mathcal{D}_o(\bar{M}, \bar{\alpha})=\Big\{(d_1,d_2): ~d_i=\lim_{\rho_{ii}\to \infty}\frac{R_i}{\log(\rho_{ii})}, i\in\{1,2\}
~\textrm{such that}~(R_1,R_2)\in \mathcal{C}(\mathcal{H},\bar{\alpha}) \Big\}. \label{def_GDOF}
\end{IEEEeqnarray}
\end{defn}

Since the capacity region of a MIMO IC is not known and a constant number of bits is insignificant in the GDoF analysis, to derive the GDoF region we shall use the constant-gap-to-capacity result found by the authors in \cite{Sanjay_Varanasi_Cap_MIMO_IC_const_gap}. In particular, since a constant number of bits is insignificant in the GDoF analysis, the $\mathcal{C}(\mathcal{H},\bar{\alpha})$ in the definition of the GDoF region can be replaced by either $\mathcal{R}^u(\mathcal{H},\bar{\alpha})$ or $\mathcal{R}_a(\mathcal{H},\bar{\alpha})$ to compute the GDoF region of the MIMO IC. We state this fact as a lemma for easy further reference.
\begin{lemma}
\label{lem_alternate_def_GDOF}
The GDoF region of the MIMO IC is given as
\begin{IEEEeqnarray}{rl}
\mathcal{D}_o(\bar{M},\bar{\alpha})=\Big\{&(d_1,d_2): ~d_i=\lim_{\rho_{ii}\to \infty}\frac{R_i}{\log(\rho_{ii})}, i\in\{1,2\} ~\textrm{and}~(R_1,R_2)\in \mathcal{R}^u(\mathcal{H},\bar{\alpha})  \Big\}, \label{eq_alternate_def_GDOF}
\end{IEEEeqnarray}
where $\mathcal{R}^u(\mathcal{H},\bar{\alpha})=\mathcal{R}^u(\mathcal{H},\bar{\rho})$ is given by Lemma~\ref{lem_upper_bound}. 
\end{lemma}

\begin{proof}
From Lemma~\ref{lem_upper_bound} and \ref{lem_achievable_region} we have the following
\begin{equation}
\mathcal{R}_a(\mathcal{H},\bar{\alpha})~~\subseteq \mathcal{C}(\mathcal{H},\bar{\alpha})~~\subseteq \mathcal{R}^u(\mathcal{H},\bar{\alpha}).
\end{equation}
To obtain the desired result we use in the definition of the GDoF region of \eqref{def_GDOF}, the above set inclusions along with the fact that $n_i$'s in Lemma \ref{lem_achievable_region} are independent of $\rho$ and $\mathcal{H}$.
\end{proof}

\subsection{Asymptotic Approximations}


In the derivation of the GDoF region of the 2-user MIMO IC quantities like the sum rate upper bound on 2- and 3-user MIMO multiple-access channels (MACs) will appear frequently. Thus, in the following two lemmas, we provide asymptotic approximations up to $O(1)$ of such quantities for different $\bar{\alpha}$ and number of antennas.

\begin{lemma}
\label{lem:sum-dof-of-2user-MAC}
Let $H_1\in \mathbb{C}^{u\times u_1}$ and $H_2\in \mathbb{C}^{u\times u_2}$ are two full rank (w.p.1) channel matrices such that $H=[H_1~H_2]$ is also full rank w.p.1. Then for asymptotic $\rho$
\begin{equation}
\label{eq_lem_myMAC_approximate1}
\log\det\left(I_u+\rho^a H_1H_1^{\dagger}+\rho^b H_2H_2^{\dagger}\right)=f\left(u, (a,u_1),(b,u_2)\right)\log(\rho)+\mathcal{O}(1),
\end{equation}
where for any $u\in \mathbb{R}^+$ and $(a_i,u_i)\in \mathbb{R}^{2}$ for $i\in \{1,2\}$,
\begin{IEEEeqnarray}{l}
\label{eq_def_f}
f\left(u,(a_1,u_1),(a_2,u_2)\right)\triangleq \left\{\begin{array}{c}
\min\{u,u_1\}a_1^++\min\{(u-u_1)^+,u_2\}a_2^+,~\textrm{if}~a_1\geq a_2;\\
\min\{u,u_2\}a_2^++\min\{(u-u_2)^+,u_1\}a_1^+,~\textrm{if}~a_1< a_2.
\end{array}\right.
\end{IEEEeqnarray}
\end{lemma}

\begin{proof}
This result was proved in \cite{PBT} when $(u_1+u_2)\geq u$. The proof for the case when $(u_1+u_2)< u$ is given in Appendix~\ref{App:proof-lemma-2user-MAC}.
\end{proof}

\begin{rem}
\label{rem:channel-distribution}
If $H_1\in \mathbb{C}^{u\times u_1}$ and $H_2\in \mathbb{C}^{u\times u_2}$ are mutually independent and $H_1, H_2\in \mathcal{P}$ then $[H_1 H_2]\in \mathcal{P}$ and therefore is a full rank matrix w.p.1. That is, if $H_1$ and $H_2$ represent the two incoming channel matrices at any of the receivers in the MIMO IC then Lemma~\ref{lem:sum-dof-of-2user-MAC} holds.
\end{rem}

\begin{lemma}
\label{lem:sum-dof-of-3user-MAC}
Let $H_i\in \mathbb{C}^{u\times u_i}$ for $i=1,2,3$ be three channel matrices with statistics described at the beginning of this section, then for asymptotic $\rho$
\begin{IEEEeqnarray}{rl}
\label{eq_lem_myMAC3_approximate1}
\log\det\left(I_u+\sum_{i=1}^3\rho^{a_i} H_iH_i^{\dagger}\right)=&g\left(u,(a_1,u_1),(a_2,u_2),(a_3,u_3)\right)\log(\rho)+\mathcal{O}(1),
\end{IEEEeqnarray}
where for any $u\in \mathbb{R}^+$ and $(a_i,u_i)\in \mathbb{R}^{2}$ for $i\in \{1,2,3\}$,
\begin{IEEEeqnarray}{rl}
g\left(u,(a_1,u_1),(a_2,u_2),(a_3,u_3)\right)\triangleq\min\{u,u_{i_1}\}a_{i_1}^++&\min\{(u-u_{i_1})^+,u_{i_2}\}a_{i_2}^++\nonumber\\
\label{eq_def_g}
&\min\{(u-u_{i_1}-u_{i_2})^+,u_{i_3}\}a_{i_3}^+,
\end{IEEEeqnarray}
for $i_1,i_2,i_3\in \{1,2,3\}$ such that $a_{i_1}\geq a_{i_2}\geq a_{i_3}$.
\end{lemma}
\begin{proof}
The proof is given in Appendix~\ref{pf_lem_my_MAC3}.
\end{proof}

\begin{rem}
Suppose $a_1\geq \max\{a_2,a_3\}$ in Lemma~\ref{lem:sum-dof-of-3user-MAC}, then $g(.)$ can also be written as
\begin{IEEEeqnarray*}{rl}
g\left(u,(a_1,u_1),(a_2,u_2),(a_3,u_3)\right)=\min\{u,u_1\}a_1^+ + f\left((u-u_1)^+,(a_2,u_2),(a_3,u_3)\right).
\end{IEEEeqnarray*}
For example, $g\left(10,(.5,3),(1,4),(1.2,2)\right)= 2(1.2)+ f\left(8,(.5,3),(1,4)\right)$ $
=2.4+ 4(1)+3(.5)=7.9$.

\end{rem}

\begin{rem}
\label{rem_MAC_interpretation}
g(.) in Lemma~\ref{lem:sum-dof-of-3user-MAC} represents the sum DoFs achievable on a 3-user MIMO multiple-access channel (MAC) with $u$ antennas
at the receiver, $u_i$ antennas at the $i^{th}$ transmitter, where the SNR of the $i^{th}$ user is $\rho^{a_i}$ for $i\in \{1,2,3\}$. Similarly, $f(.)$ in Lemma~\ref{lem:sum-dof-of-2user-MAC} can be interpreted as the sum GDoF achievable on a 2-user MIMO MAC.
\end{rem}

\section{The GDoF region of the MIMO IC}
\label{sec_main_result}
Using the explicit expression for the upper bounds to the capacity region of the MIMO IC from Lemma \ref{lem_upper_bound} and using it in Lemma~\ref{lem_alternate_def_GDOF} we get the main result of this paper.
\begin{thm}
\label{thm_mainresult}
The GDoF region of $\mathcal{IC}(\bar{M},\bar{\alpha})$ is the set of DoF tuples $(d_1,d_2)$, denoted by $\mathcal{D}_o(\bar{M},\bar{\alpha})$, where $d_i\in \mathbb{R}^{+}$ for $i=1,2$ satisfy the following conditions:
\begin{IEEEeqnarray*}{rl}
\label{eq_gdofr1}
d_1\leq &\min\{M_1, N_1\};\\
d_2\leq &\min\{M_2, N_2\};\\
\label{eq_gdofr3}
(d_1+\alpha_{22}d_2)\leq & f\left(N_2, (\alpha_{12},M_1),(\alpha_{22},M_2)\right)+ f\left(N_1, (\beta_{12},m_{12}),(\alpha_{11},(M_1-N_2)^+)\right);\\
\label{eq_gdofr4}
(d_1+\alpha_{22}d_2)\leq & f\left(N_1, (\alpha_{21},M_2),(\alpha_{11},M_1)\right)+f\left(N_2, (\beta_{21},m_{21}),(\alpha_{22},(M_2-N_1)^+)\right);\\
(d_1+\alpha_{22}d_2)\leq & g\left(N_1,(\alpha_{21},M_2),(\beta_{12},m_{12}),(1,(M_1-N_2)^+)\right)+\nonumber \\
\label{eq_gdofr5}
&~~~~~~~~~~~~g\left(N_2,(\alpha_{12},M_1),(\beta_{21},m_{21}),(\alpha_{22},(M_2-N_1)^+)\right);\\
(2d_1+\alpha_{22}d_2)\leq & f\left(N_1, (\alpha_{21},M_2),(\alpha_{11},M_1)\right)+f\left(N_1, (\beta_{12},m_{12}),(\alpha_{11},(M_1-N_2)^+)\right)+\nonumber\\
\label{eq_gdofr6}
&~~~~~~~~~~~~g\left(N_2,(\alpha_{12},M_1),(\beta_{21},m_{21}),(\alpha_{22},(M_2-N_1)^+)\right);\\
(d_1+2\alpha_{22}d_2)\leq & f\left(M_2, (\alpha_{21},N_1),(\alpha_{22},N_2)\right)+f\left(N_2, (\beta_{21},m_{21}),(\alpha_{22},(M_2-N_1)^+)\right)+\nonumber\\
\label{eq_gdofr7}
&~~~~~~~~~~~~g\left(N_1,(\alpha_{21},M_2),(\beta_{12},m_{12}),(1,(M_1-N_2)^+)\right),
\end{IEEEeqnarray*}
where $\beta_{ij}=(\alpha_{ii}-\alpha_{ij})^+$, functions $f(.,.,.)$ and g(.,.,.,.) are as defined in equation \eqref{eq_def_f} and \eqref{eq_def_g}, respectively, for $i\neq j\in \{1,2\}$ and and $m_{ij} \triangleq \min\{M_i,N_j\}$ as defined before.
\end{thm}

\begin{proof}[Proof of Theorem~\ref{thm_mainresult}(Outline)]
From Lemma~\ref{lem_alternate_def_GDOF} we see that the GDoF region of the 2-user MIMO IC is simply the scaled version of the rate region $\mathcal{R}^u(\mathcal{H},\bar{\alpha})$. Thus to evaluate the GDoF region we simply need to scale all the terms in each of the equations of $\mathcal{R}^u(\mathcal{H},\bar{\alpha})$. For example, consider the third bound in equation ~\eqref{eq_bound3}
\begin{IEEEeqnarray*}{l}
R_1+R_2\leq I_{b3}.
\end{IEEEeqnarray*}
Dividing both sides by $\log(\rho_{11})$ and taking the limit we get
\begin{IEEEeqnarray*}{l}
\lim_{\rho_{11}\to \infty}\frac{R_1+R_2}{\log(\rho_{11})}\leq  \lim_{\rho_{11}\to \infty}\frac{I_{b3}}{\log(\rho_{11})};\\
\Rightarrow~\lim_{\rho \to \infty}\left\{ \frac{R_1}{\log(\rho)}+\frac{\alpha_{22} R_2}{\alpha_{22}\log(\rho)}\right\}\leq  \lim_{\rho\to \infty}\frac{I_{b3}}{\log(\rho)};\\
\Rightarrow~\lim_{\rho \to \infty} \frac{R_1}{\log(\rho)}+\lim_{\rho_{22} \to \infty}\frac{\alpha_{22} R_2}{\log(\rho_{22})}\leq  \lim_{\rho\to \infty}\frac{I_{b3}}{\log(\rho)};\\
\Rightarrow~d_1+\alpha_{22} d_2 \leq  \lim_{\rho\to \infty}\frac{I_{b3}}{\log(\rho)}.
\end{IEEEeqnarray*}
Following the same steps for other bounds we get
\begin{IEEEeqnarray}{rl}
\label{eq_outline_gdof_b1}
\mathcal{D}_o(\bar{M},\bar{\alpha})=\Bigg\{(d_1,d_2):~d_1\leq &\lim_{\rho\to \infty}\frac{I_{b1}}{\log(\rho)};\\
\label{eq_outline_gdof_b2}
d_2\leq &\lim_{\rho\to \infty}\frac{I_{b2}}{\log(\rho)};\\
\label{eq_outline_gdof_b3}
d_1+\alpha_{22}d_2\leq & \lim_{\rho\to \infty}\frac{I_{b3}}{\log(\rho)};\\
\label{eq_outline_gdof_b4}
d_1+\alpha_{22}d_2\leq & \lim_{\rho\to \infty}\frac{I_{b4}}{\log(\rho)};\\
\label{eq_outline_gdof_b5}
d_1+\alpha_{22}d_2\leq & \lim_{\rho\to \infty}\frac{I_{b5}}{\log(\rho)};\\
\label{eq_outline_gdof_b6}
2d_1+\alpha_{22}d_2\leq & \lim_{\rho\to \infty}\frac{I_{b6}}{\log(\rho)};\\
\label{eq_outline_gdof_b7}
d_1+2\alpha_{22}d_2\leq & \lim_{\rho\to \infty}\frac{I_{b7}}{\log(\rho)};\Bigg\},
\end{IEEEeqnarray}
To prove the theorem we have to evaluate the right hand side limits, which can be done by finding the asymptotic approximations of the different $I_{bi}$s of Lemma~\ref{lem_upper_bound} using Lemmas \ref{lem:sum-dof-of-2user-MAC} and \ref{lem:sum-dof-of-3user-MAC}. The detailed proof is given in Appendix~\ref{proof_thm_mainresult}.
\end{proof}

The above theorem is specialized next to the SISO IC (by putting $M_1=M_2=N_1=N_2=1$) in the following corollary, yielding its GDoF region.
\begin{cor}
\label{gdof-siso-ic}
The GDoF region of the SISO IC is given as
\begin{subequations}
\label{eq:GDoF-region-of-SISO-IC}
\begin{align}
\mathcal{D}_{\textrm{SISO}}=\Big\{(d_1,d_2):d_1\leq &1;\\
d_2\leq &1;\\
(d_1+\alpha_{22} d_2)\leq & \max\{\alpha_{22},\alpha_{12}\}+(1-\alpha_{12})^+;\\
(d_1+\alpha_{22} d_2)\leq & \max\{1,\alpha_{21}\}+(\alpha_{22}-\alpha_{21})^+;\\
(d_1+\alpha_{22} d_2)\leq & \max\{\alpha_{21}, (1-\alpha_{12})^+\}+\max\{\alpha_{12}, (\alpha_{22}-\alpha_{21})^+\};\\
(2d_1+\alpha_{22} d_2)\leq & \max\{1,\alpha_{21}\}+(1-\alpha_{12})^+\max\{\alpha_{12}, (\alpha_{22}-\alpha_{21})^+\};\\
(d_1+2\alpha_{22} d_2)\leq & \max\{\alpha_{22},\alpha_{12}\}+\max\{\alpha_{21}, (1-\alpha_{12})^+\}+(\alpha_{22}-\alpha_{21})^+\Big\}
\end{align}
\end{subequations}
\end{cor}

\begin{rem}
The region of Corollary \ref{gdof-siso-ic} provides a single unified formula for the GDoF region of the SISO IC for all interference regimes. It can be specialized to obtain GDoF regions for different interference regimes given in Section V of \cite{ETW1} such as weak interference, mixed interference, strong interference, etc., separate formulas for each of which are given therein. For instance, in the {\em weak} interference regime defined by $\alpha_{12}\leq 1$ and $\alpha_{21}\leq \alpha_{22}$, we have
\begin{eqnarray*}
\max\{\alpha_{22},\alpha_{12}\}+(1-\alpha_{12})^+=\max\{\alpha_{22},\alpha_{12}\}+(1-\alpha_{12})=1+(\alpha_{22}-\alpha_{12})^+;\\
\max\{1,\alpha_{21}\}+(\alpha_{22}-\alpha_{21})^+=\max\{1,\alpha_{21}\}+(\alpha_{22}-\alpha_{21})=\alpha_{22}+(1-\alpha_{21})^+.
\end{eqnarray*}
Substituting these identities in equation \eqref{eq:GDoF-region-of-SISO-IC} we recover equation (78) of \cite{ETW1} which represents the GDoF region of the SISO IC in the weak interference regime (note that $\alpha_{22}, \alpha_{21} $ and $\alpha_{12}$ are denoted as $\alpha_1, \alpha_2$ and $\alpha_3$ in \cite{ETW1}). Similarly for instance, equations (82) and (84) of \cite{ETW1} can be recovered for the mixed and strong interference channels, respectively, by simplifying the result of Corollary \ref{gdof-siso-ic} according to the defining conditions on the $\alpha$'s for those regimes.
\end{rem}

\begin{rem}
The conventional DoF region of the MIMO IC obtained in \cite{JFak} can also be recovered from Theorem~\ref{thm_mainresult} by putting $\alpha_{ij}=1$, for $1\leq i,j\leq 2$ in Theorem~\ref{thm_mainresult} and simplifying the different bounds. Consequently, we get
\begin{subequations}
\label{eq:DoF-region-of-MIMO-IC}
\begin{align}
\label{eq:DoF-region-of-MIMO-IC-a}
\mathcal{D}_{\textrm{DoF}}=\Big\{(d_1,d_2):d_1\leq &\min\{M_1,N_1\};\\
\label{eq:DoF-region-of-MIMO-IC-b}
d_2\leq &\min\{M_1,N_1\};\\
\label{eq:DoF-region-of-MIMO-IC-c}
(d_1+d_2)\leq & (N_2\land (M_1+M_2))+N_1\land (M_1-N_2)^+;\\
\label{eq:DoF-region-of-MIMO-IC-d}
(d_1+d_2)\leq & (N_1\land (M_1+M_2))+N_2\land (M_2-N_1)^+;\\
\label{eq:DoF-region-of-MIMO-IC-e}
(d_1+d_2)\leq & (N_1\land M_2)+ ((N_1-M_2)^+\land (M_1-N_2)^+) + \nonumber\\
&~~~(N_2\land M_1)+ ((M_1-N_2)^+\land(M_2-N_1)^+);\\
\label{eq:DoF-region-of-MIMO-IC-f}
(2d_1+d_2)\leq & (N_1\land (M_1+M_2))+N_1\land (M_1-N_2)^++\nonumber\\
&~~~(N_2\land M_1)+ ((M_1-N_2)^+\land(M_2-N_1)^+);\\
\label{eq:DoF-region-of-MIMO-IC-g}
(d_1+2d_2)\leq & (N_2\land (M_1+M_2))+ N_2\land (M_2-N_1)^++ \nonumber\\
&~~~(N_1\land M_2)+ ((N_1-M_2)^+\land (M_1-N_2)^+)\Big\}
\end{align}
\end{subequations}

\begin{cor}[The main result of \cite{JFak}]
\label{claim:DoF-region}
The DoF region of the 2-user MIMO IC is given as
\begin{IEEEeqnarray*}{rl}
\mathcal{D}_{\textrm{DoF}}=\Big\{(d_1,d_2):d_1\leq &\min\{M_1,N_1\};\\
d_2\leq &\min\{M_2,N_2\};\\
(d_1+d_2)\leq & \min\{(M_1+M_2),(N_1+N_2),\max(M_1,N_2),\max(M_1,N_2)\}\Big\}.
\end{IEEEeqnarray*}
\end{cor}

\begin{proof}
We obtain this result starting from equation \eqref{eq:DoF-region-of-MIMO-IC}. Let us consider the sum bound of equation~\eqref{eq:DoF-region-of-MIMO-IC-c},
\begin{IEEEeqnarray*}{rl}
(d_1+d_2)\leq & (N_2\land (M_1+M_2))+N_1\land (M_1-N_2)^+,\\
=&\{N_2\land (M_1+M_2)\}1(N_2\geq M_1)+\{N_2+ (N_1\land (M_1-N_2))\}1(N_2< M_1),\\
=&\{N_2\land (M_1+M_2)\}1(N_2\geq M_1)+\{(N_2+N_1)\land M_1\}1(N_2< M_1),\\
=&\min \{(M_1+M_2),(N_1+N_2),\max(N_2,M_1)\}.
\end{IEEEeqnarray*}
Similarly, simplifying the bound in equation \eqref{eq:DoF-region-of-MIMO-IC-d} it can be shown that
\begin{IEEEeqnarray*}{rl}
(d_1+d_2)\leq & \min \{(M_1+M_2),(N_1+N_2),\max(N_1,M_2)\}.
\end{IEEEeqnarray*}
Moreover, in Appendix~\ref{App:proof-of-DoF-region} it will be shown that the bounds in equations \eqref{eq:DoF-region-of-MIMO-IC-e}-\eqref{eq:DoF-region-of-MIMO-IC-g} are looser than those in equation \eqref{eq:DoF-region-of-MIMO-IC-a}-\eqref{eq:DoF-region-of-MIMO-IC-d}. Finally, combining the simplified forms of the 3-rd and 4-th bounds above, the  claim is proved.
\end{proof}

\end{rem}

\begin{rem}
The GDoF region of the 2-user MIMO multiple-access channel (MAC) with an arbitrary number of antennas at the three terminals can also be found as a by-product of the analysis of the MIMO IC. This is detailed in Appendix \ref{App:GDoF:MIMO-MAC}.
\end{rem}

\begin{rem}[Interpretation of the different bounds]
We know that the GDoF optimal coding scheme divides each user's message into two sub-messages. Let the DoFs of the private and the public messages of user $i$ be denoted by $d_{ip}$ and $d_{ic}$, respectively. Note that $H_{12}$ has a $(M_1-N_2)^+$-dimensional null space along which $Tx_1$ can send private information to its desired receiver at an SNR of $\rho^{\alpha_{11}}$. Along the remaining $m_{12}$ dimensions $Tx_1$ can send private information only at a power level of $\rho^{-\alpha_{12}}$ which reaches $Rx_1$ at a power level of $\rho^{(\alpha_{11}-\alpha_{12})^+}$. Thus with respect to the private information of $Tx_1$, $Rx_1$ is a MAC with 2 virtual transmitters having SNRs $\rho^{\alpha_{11}}$ and $\rho^{(\alpha_{11}-\alpha_{12})^+}$ and $(M_1-N_2)^+$ and $m_{12}$ transmit antennas, respectively. Hence, from Lemma~\ref{lem:sum-dof-of-2user-MAC}, we have
\begin{equation*}
d_{1p}\leq f\left(N_1, ((\alpha_{11}-\alpha_{12})^+,m_{12}),(\alpha_{11},(M_1-N_2)^+)\right).
\end{equation*}
On the other hand, since $d_{1c}$ is decoded at $Rx_2$, $Rx_2$ is a MAC receiver with respect to $W_1$ (having an SNR of $\rho^{\alpha_{12}}$ and $M_1$ transmit antennas) and $X_2$ (having an SNR of $\rho^{\alpha_{22}}$) and from Lemma~\ref{lem:sum-dof-of-2user-MAC} (recall Remark~\ref{rem_MAC_interpretation}) we have
\begin{equation*}
    (d_{1c}+\alpha_{22} d_2)\leq f\left(N_2, (\alpha_{12},M_1),(\alpha_{22},M_2)\right).
\end{equation*}
Combining the above two equations we get the $3^{rd}$ bound of the GDoF region. The $4^{th}$ bound can be similarly interpreted just by interchanging the roles of $Rx_1$ and $Rx_2$. As explained above, the two parts of the private message of $Tx_1$ can be thought of as two virtual users to the MAC receiver $Rx_1$; in addition to them, $Tx_2$ can send a maximum of $m_{21}\alpha_{21}$ public DoFs to $Rx_1$ through $W_2$, which can be interpreted as the $3^{rd}$ virtual user (with SNR $\rho^{\alpha_{21}}$ and $m_{21}$ transmit antennas) to the MAC receiver at $Rx_1$, and therefore, Lemma \ref{lem:sum-dof-of-3user-MAC} provides the following sum DoF uppper bound
\begin{equation*}
    (d_{1p}+\alpha_{22} d_{2c})\leq g\left(N_1,(\alpha_{21},M_2),(\beta_{12},m_{12}),(1,(M_1-N_2)^+)\right).
\end{equation*}
A similar consideration regarding the DoFs decodable at $Rx_2$ gives
\begin{equation*}
    (d_{1c}+\alpha_{22} d_{2p})\leq g\left(N_2,(\alpha_{12},M_1),(\beta_{21},m_{21}),(\alpha_{22},(M_2-N_1)^+)\right).
\end{equation*}
Combining the last two equations we get the $5^{th}$ bound of Theorem~\ref{thm_mainresult}. The other two bounds of the theorem can be similarly interpreted.
\end{rem}

\begin{rem}[DoF-Split]
The distributions and power levels of the codewords that encode each user's private and public messages are specified in Section \ref{def_coding_scheme}; however, to completely specify the encoding scheme we further need to specify the DoFs carried by the private and public messages of each user, which are denoted by $d_{ip}$ and $d_{ic}$, respectively, for $i=1,2$. The lemma below completes the specification of the GDoF optimal coding scheme by providing a set of 4-tuples, $\mathcal{G}(\bar{M},\bar{\alpha})=\{(d_{1c}, d_{1p}, d_{2c}, d_{2p})\}$, which is achievable on the 2-user MIMO IC by the GDoF optimal coding scheme of Section \ref{def_coding_scheme}. The $\mathcal{G}(\bar{M},\bar{\alpha}) $ region has the property that, for any $(d_1,d_2)\in \mathcal{D}_o(\bar{M},\bar{\alpha})$ (which is specified in Theorem \ref{thm_mainresult}), there exists an $(d_{1c}, d_{1p}, d_{2c}, d_{2p})\in \mathcal{G}(\bar{M},\bar{\alpha})$ such that $(d_{ip}+d_{ic})=d_i$ for $i=1,2$.
\end{rem}

\begin{lemma}
\label{lem:DoF-split}
The DoF pair $(d_1,d_2)\in \mathcal{D}_o(\bar{M},\bar{\alpha})$ only if there exists a 4-tuple $(d_{1c}, d_{1p},d_{2c}, d_{2p}) \in \mathcal{G}(\bar{M},\bar{\alpha})$ such that $d_i=(d_{ic}+d_{ip})$ for $i=1,2$, where $\mathcal{G}(\bar{M},\bar{\alpha})=\mathcal{G}_1(\bar{M},\bar{\alpha})\cap \mathcal{G}_2(\bar{M},\bar{\alpha})$ with $ \mathcal{G}_1(\bar{M},\bar{\alpha}) $ defined below (and with $\mathcal{G}_2(\bar{M},\bar{\alpha})$ obtained by interchanging the indexes 1 and 2 in the expression for $\mathcal{G}_1(\bar{M},\bar{\alpha})$),
\begin{subequations}
\label{eq:GDoF-sub-region}
\begin{align}
\mathcal{G}_1(\bar{M},\bar{\alpha})=\Big\{(d_{1p},d_{1c},d_{2c}):~\alpha_{11}d_{1p}\leq & f\left(N_1, (\beta_{12},m_{12}),(\alpha_{11},(M_1-N_2)^+)\right)\alpha_{11};\\
\alpha_{11}d_{1c}\leq &\min\{N_1,M_1, N_2\}\alpha_{11};\\
\alpha_{22}d_{2c}\leq &\min\{N_1,M_2\}\alpha_{21};\\
\alpha_{11}(d_{1p}+d_{1c})\leq & \min\{M_1,N_1\}\alpha_{11};\\
(\alpha_{11}d_{1p}+\alpha_{22}d_{2c})\leq & g\left(N_1,(\alpha_{21},M_2),(\beta_{12},m_{12}),(1,(M_1-N_2)^+)\right);\\
(\alpha_{11}d_{1c}+\alpha_{22}d_{2c})\leq & f\left(N_1, (\alpha_{21},M_2),(\alpha_{11},m_{12})\right);\\
(\alpha_{11}d_{1p}+\alpha_{11}d_{1c}+\alpha_{22}d_{2c})\leq & f\left(N_1, (\alpha_{21},M_2),(\alpha_{11},M_1)\right);\Big\}
\end{align}
\end{subequations}
with $\beta_{12}=(\alpha_{11}-\alpha_{12})^+$ and functions $f(.,.,.)$ and g(.,.,.,.) are as defined in equation \eqref{eq_def_f} and \eqref{eq_def_g}, respectively.
\end{lemma}

\begin{proof}
The proof is given in Appendix~\ref{app:lem-dof-split}.
\end{proof}

\section{The Symmetric GDoF region of the $(M,N,M,N)$ MIMO IC}
\label{sec:reciprocity-of-GDoF-region}
Suppose the roles of the transmitters and receivers of the MIMO IC $\mathcal{IC}(\mathcal{H},\bar{\alpha})$ are interchanged. In the notations defined in Section \ref{sec_channel_model_and_preliminaries}, this resulting IC (hereafter referred to as the ``reciprocal" channel) can be denoted by $\mathcal{IC}(\mathcal{H}^r,\bar{\alpha}^r)$, where $\mathcal{H}^r=\{H_{11}^T,H_{21}^T,H_{12}^T, H_{22}^T\}$ and $\bar{\alpha}^r=[\alpha_{11},\alpha_{21},\alpha_{12},\alpha_{22}]$. Clearly, $\mathcal{D}_o(\bar{M}^r,\bar{\alpha}^r)$ denotes the GDoF region of the reciprocal channel where $\bar{M}^r=(N_1,M_1,N_2,M_2)$.

\begin{cor}[Reciprocity of the GDoF region]
\label{cor:reciprocity}
The GDoF region of the MIMO IC is same as that of its {\it reciprocal} channel i.e.,
\begin{equation*}\label{cor_reciprocity_GDOF}
\mathcal{D}_o(\bar{M},\bar{\alpha})=\mathcal{D}_o(\bar{M}^r,\bar{\alpha}^r).
\end{equation*}
\end{cor}

\begin{proof}
It was proved in \cite{Sanjay_Varanasi_Cap_MIMO_IC_const_gap} that the capacity region of a 2-user MIMO IC and its {\it reciprocal} channel are within a constant (independent of $\rho$) number of bits to each other (see Lemma~6 of \cite{Sanjay_Varanasi_Cap_MIMO_IC_const_gap}). The corollary is easily proved by using this result in the definition of the GDoF region of the IC in equation \eqref{def_GDOF}, which states that the GDoF regions of two channels with capacity regions differing by only a constant number of bits are the same.
\end{proof}

In other words, the GDoF region of the channel does not change if the roles of the transmitters and the receivers are interchanged. Note that this is a more general result than the reciprocity of the conventional DoF region proved in \cite{JFak}. In what follows, we define the symmetric GDoF metric.

\begin{defn}[Symmetric GDoF]
\label{def:symmetric-GDoF}
Let $\mathcal{C}_s(\alpha)=\sup (R_1+R_2)$ with $(R_1,R_2)\in \mathcal{C}(\mathcal{H},\bar{\alpha})$ where $\bar{\alpha}=[1,\alpha,\alpha,1]$ and $\sup \mathcal{A}$ represents the supremum of the set of elements in $\mathcal{A}$. Then the {\em symmetric GDoF} of the channel, denoted by $d_s$, is defined as
\begin{equation*}
    d_s\triangleq \lim_{\rho \to \infty} \frac{\mathcal{C}_s(\alpha)}{2\log(\rho)}.
\end{equation*}
\end{defn}
It is clear from Definition \ref{def:GDoF-region} and the above equation that
\begin{equation}
\label{eq:symmetric-GDoF}
    d_s ~=~\frac{\sup_{\mathcal{D}_o(\bar{M},\bar{\alpha})} (d_1+d_2)}{2}.
\end{equation}

The authors in \cite{PBT} found the symmetric GDoF of the MIMO IC with an equal number of antennas at the two transmitters and an equal number of antennas at the receiver and with more receive than transmit antennas, i.e., for an $(M,N,M,N)$ IC with $M\leq N$ and $\bar{\alpha}=[1,\alpha,\alpha,1]$. They found that
\begin{equation}
d_{s}\leq \min\{M,\hat{D}(\alpha)\}
\end{equation}
where
\begin{IEEEeqnarray}{l}
\hat{D}(\alpha)=\left\{\begin{array}{ll}
M-(2M-N)\alpha, &0\leq \alpha<\frac{1}{2};\\
(N-M)+(2M-N)\alpha, &\frac{1}{2}\leq\alpha\leq \frac{2}{3};\\
M-\frac{\alpha}{2}(2M-N),  &\frac{2}{3}\leq\alpha\leq 1;\\
\frac{N}{2}+\frac{M}{2}(\alpha-1), &1\leq \alpha.
\end{array}\right.
\label{eq:dalpha}
\end{IEEEeqnarray}

It can be verified that the above result of \cite{PBT} can be recovered by putting  $N_1=N_2=N, M_1=M_2=M$ and $\bar{\alpha}=[1,\alpha,\alpha,1]$, in Theorem~\ref{thm_mainresult} of this paper (the details are left to the reader). The result of \cite{PBT} is however only valid for $M\leq N$ and does not extend to the $M>N$ case. Specializing Theorem~\ref{thm_mainresult} for $N_1=N_2=N, M_1=M_2=M>N$ and $\bar{\alpha}=[1,\alpha,\alpha,1]$ we get the following.
\begin{cor}
\label{cor_gdof_symmetric_MgeqN}
The symmetric DoF ($d_{s}=d_1=d_2$) of a $(M,N,M,N)$ IC with $M>N$ and $\bar{\alpha}=[1,\alpha,\alpha,1]$ is given by
\begin{equation*}
d_{s}\leq \min \{N, D(\alpha)\}
\end{equation*}
where $D(\alpha)$ is given as
\begin{IEEEeqnarray}{l}
\label{eq_cor_MgN_case}
D(\alpha)=\left\{\begin{array}{ll}
N-(2N-M)\alpha, &0\leq \alpha<\frac{1}{2};\\
(M-N)+(2N-M)\alpha, &\frac{1}{2}\leq\alpha\leq \frac{2}{3};\\
N-\frac{\alpha}{2}(2N-M),  &\frac{2}{3}\leq\alpha\leq 1;\\
\frac{M}{2}+\frac{N}{2}(\alpha-1), &1\leq \alpha.
\end{array}\right.
\end{IEEEeqnarray}
\end{cor}

\begin{rem}
The above formula is the same as the one given by \eqref{eq:dalpha} with $M$ and $N$ interchanged. This is in accordance with the reciprocity result in Corollary \ref{cor_reciprocity_GDOF}. In other words, $d_{s}$ of the $(M,N,M,N)$ IC with $M>N$ for a given $\frac{M}{N}=r$ is the same as that of the GDoF of a MIMO IC with $M\leq N$ and $\frac{M}{N}=\frac{1}{r}$.
\end{rem}

\begin{rem}
It must be noted that the achievable schemes on the two channels are entirely different. While for $M\leq N$ the coding scheme need not depend on the channel matrices at the transmitters (see the achievability scheme of \cite{PBT}), for $M>N$ the covariance matrices are necessarily functions of the channel matrices. Hence, a naive extension of the scheme of \cite{PBT} to the case of $M>N$ is not GDoF optimal. In fact, such a scheme  wouldn't even be DoF optimal because while on a MIMO IC with $M\leq N$, receive zero-forcing is sufficient to achieve the DoF region, for $M>N$, knowledge of channel state information at the transmitters (CSIT) is necessary to achieve DoF-optimal performance \cite{Huang_Jafar,Vaze_Varanasi,ZG} such as through transmit zero-forcing beamforming \cite{JFak}.
\end{rem}

\begin{rem}[A scheme that ignores CSIT:]
The GDoF optimal coding scheme of \cite{PBT} does not utilize any CSIT. The approach in \cite{PBT} was to divide the range of $\alpha $ into the five regimes delineated in the SISO IC case in \cite{ETW1} and employ the main idea of the achievable schemes that are known to be GDoF-optimal in the SISO case (e.g., treat interference as noise in the very weak interference regime; set the power level of private messages so they arrive at the noise level at the unintended receiver in the moderate and weak interference regimes following the prescription of \cite{ETW1}, send only common messages in the strong and very strong interference regimes). Hence, this coding scheme effectively only employs signal level interference alignment without any form of transmit beamforming. While on an $(M,N,M,N)$ IC with $M\leq N$ beamforming is not necessary because neither of the cross-links have a null space, it is so when $M>N$. Therefore, the coding scheme of \cite{PBT} when applied naively to the MIMO IC with $M>N$ cannot achieve the fundamental GDoF region of the channel (e.g., see Fig. \ref{figure_gdof_MgeqN_channel} and Fig. \ref{figure_gdof_PBT_scheme}).
\end{rem}

\begin{rem}[Treating interference as noise (TIN):] On a 2-user MIMO IC with $\alpha\leq \frac{1}{2}$, $M_1=M_2=M$ and $N_1=N_2=N$, treating interference as noise (TIN) gives the following achievable region
\begin{IEEEeqnarray}{rl}
\Big\{ (R_1,R_2): & R_1\leq \log\det\left(I_{M}+\rho H_{11}^\dagger\left(I_{N}+\rho^{\alpha}H_{21}H_{21}^\dagger\right)^{-1}H_{11}\right),\nonumber \\
\label{eq_TIN_formula}
& R_2\leq \log\det\left(I_{M}+\rho H_{22}^\dagger\left(I_{N}+\rho^{\alpha}H_{12}H_{12}^\dagger\right)^{-1}H_{22}\right)\Big\}.
\end{IEEEeqnarray}
The corresponding GDoF region is given as
\begin{IEEEeqnarray*}{rl}
\Big\{ (d_1,d_2): & d_1\leq f\left(N,(\alpha,M),(1,M)\right)-\min \{M,N\}\alpha=N(1-\alpha),\nonumber \\
& d_2\leq f\left(N,(\alpha,M),(1,M)\right)-\min \{M,N\}\alpha=N(1-\alpha)\Big\}.
\end{IEEEeqnarray*}
\end{rem}

Fig. \ref{figure_gdof_MgeqN_channel} and Fig. \ref{figure_gdof_PBT_scheme} show the GDoF achievable by the GDoF optimal coding scheme of this paper in comparison with the coding scheme used in \cite{PBT} and the TIN scheme (whose GDoF is denoted by the dashed red line). Comparing the GDoF curves of the three schemes, it is clear that both TIN and the No-CSIT coding scheme of \cite{PBT} fail to achieve the fundamental GDoF of the $(M,N,M,N)$ IC when $M>N$.

\begin{figure}[htp]
  \begin{center}
    \subfigure[The symmetric GDoF of the channel.]{\label{figure_gdof_MgeqN_channel}\includegraphics[scale=.5]{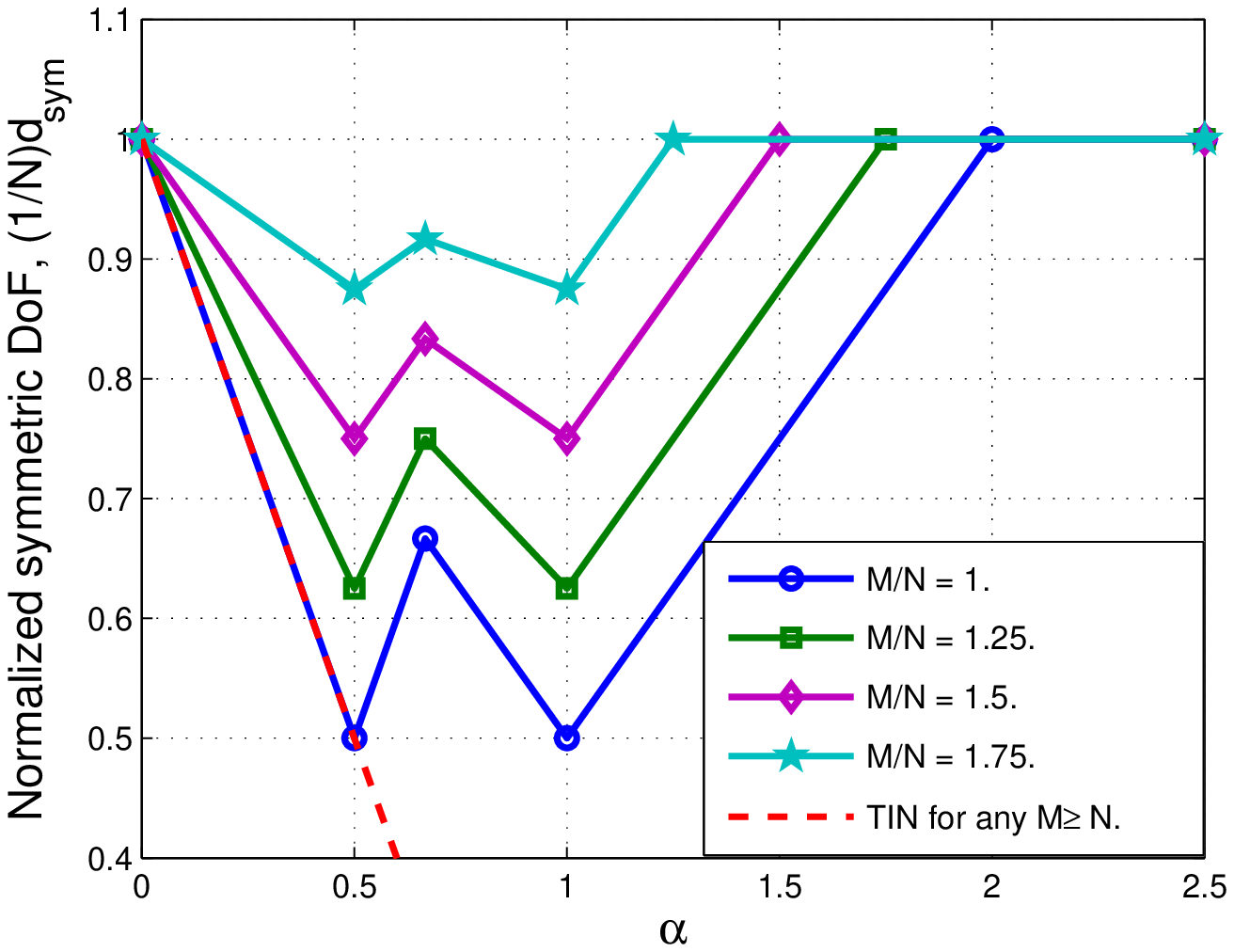}}
    \subfigure[Symmetric GDoF achievable by the coding scheme of \cite{PBT}.]{\label{figure_gdof_PBT_scheme}\includegraphics[scale=0.5]{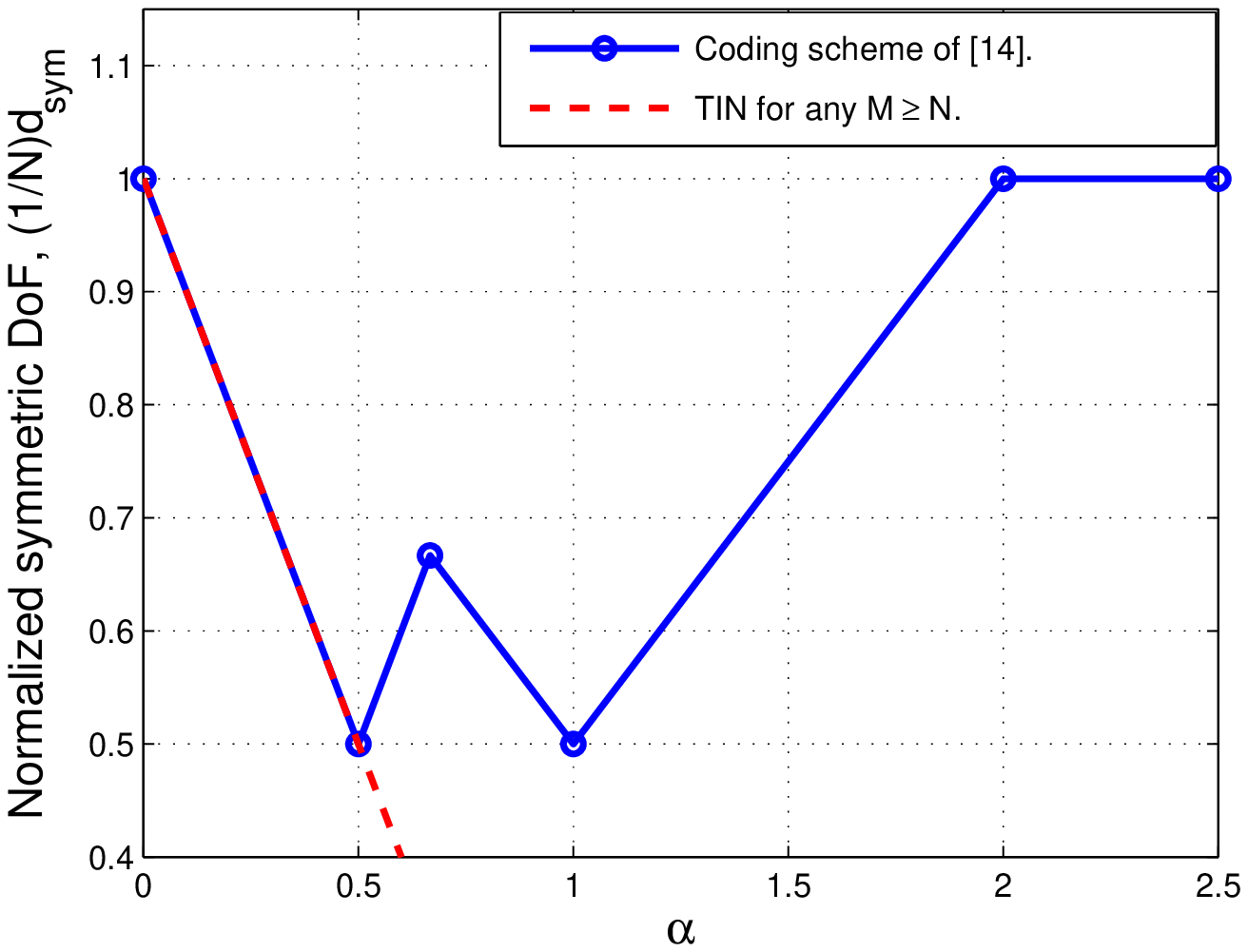}}
  \end{center}
\caption{Symmetric GDoF of the $(M,N,M,N)$ IC, with $M\geq N$. }
\end{figure}

\section{How does $\mathcal{HK}(\{K_{1u},K_{1w},K_{2u},K_{2w}\})$ achieve GDoF optimality?}
\label{sec:explicit-examples}
In what follows, we shall explain how the $\mathcal{HK}(\{K_{1u},K_{1w},K_{2u},K_{2w}\})$ scheme achieves the fundamental GDoF region of the MIMO IC through some examples. The encoding scheme is completely specified in Subsection \ref{def_coding_scheme} and Lemma~\ref{lem:DoF-split}. As explained in Subsection~\ref{def_coding_scheme}, in the GDoF optimal coding scheme the private and public messages of each user are essentially a weighted sum of several independent streams of information, each stream directed along a beam which is dependent on the channel matrix of the cross link emerging from the corresponding transmitter. The direction of these beams and their weights are chosen in such a manner (e.g., see equations \eqref{eq:expansion-of-kiu}-\eqref{eq_structure_of_streams}) that the effective covariance matrix of the overall codeword corresponding to each of the message is as given by equation~\eqref{eq_power_split}\footnote{Instead of sending independent streams of information (which is without loss of GDoF optimality), if coding is also done across different streams, it is possible to achieve a larger error exponents.}. 
As for decoding, it is clear that with respect to $U_i$, $W_i$ and $W_j$, $Rx_i$ sees a MAC channel for $i\neq j\in \{1,2\}$ and for any $(d_{1p},d_{1c},d_{2p},d_{2c})$-tuple belonging to the achievable region (see Lemma~\ref{lem:DoF-split} in Appendix~\ref{app:lem-dof-split}) $Rx_i$ can decode $U_i$, $W_i$ and $W_j$ with probability of error going to zero. Therefore, any decoding scheme which is capacity optimal on a MAC will be GDoF optimal for the MIMO IC if each receiver tries to decode the 2 public messages and its own private message while treating the other private message as noise.

\begin{figure}[htp]
  \begin{center}
    \subfigure[$\bar{\alpha}=(1,\frac{3}{5}, \frac{3}{5},1)$.]{\label{fig_gdof_achievable_1}\includegraphics[scale=.3]{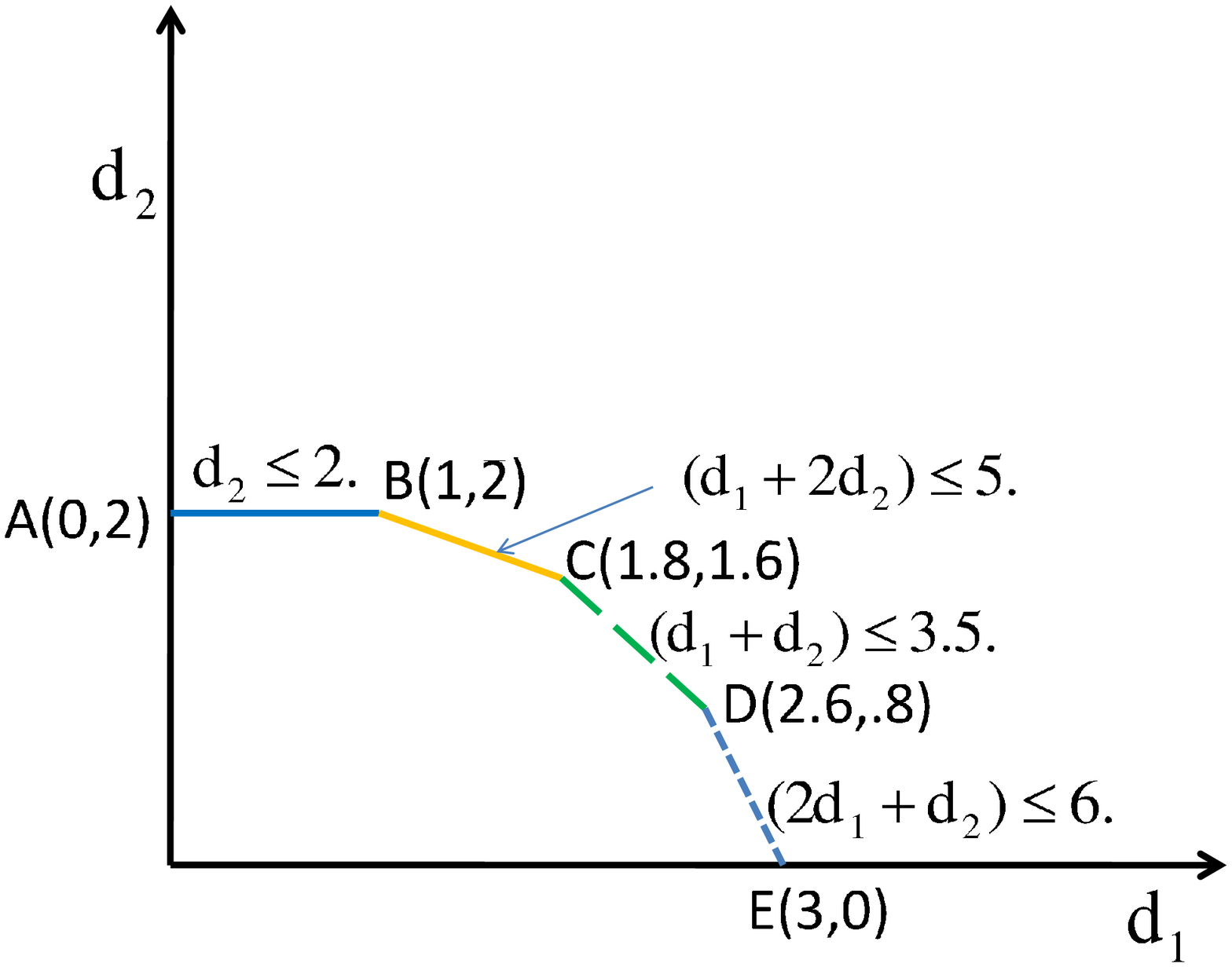}}
    \subfigure[$\bar{\alpha}=(1,\frac{1}{4}, \frac{5}{4},1)$.]{\label{fig_gdof_achievable_2}\includegraphics[scale=0.3]{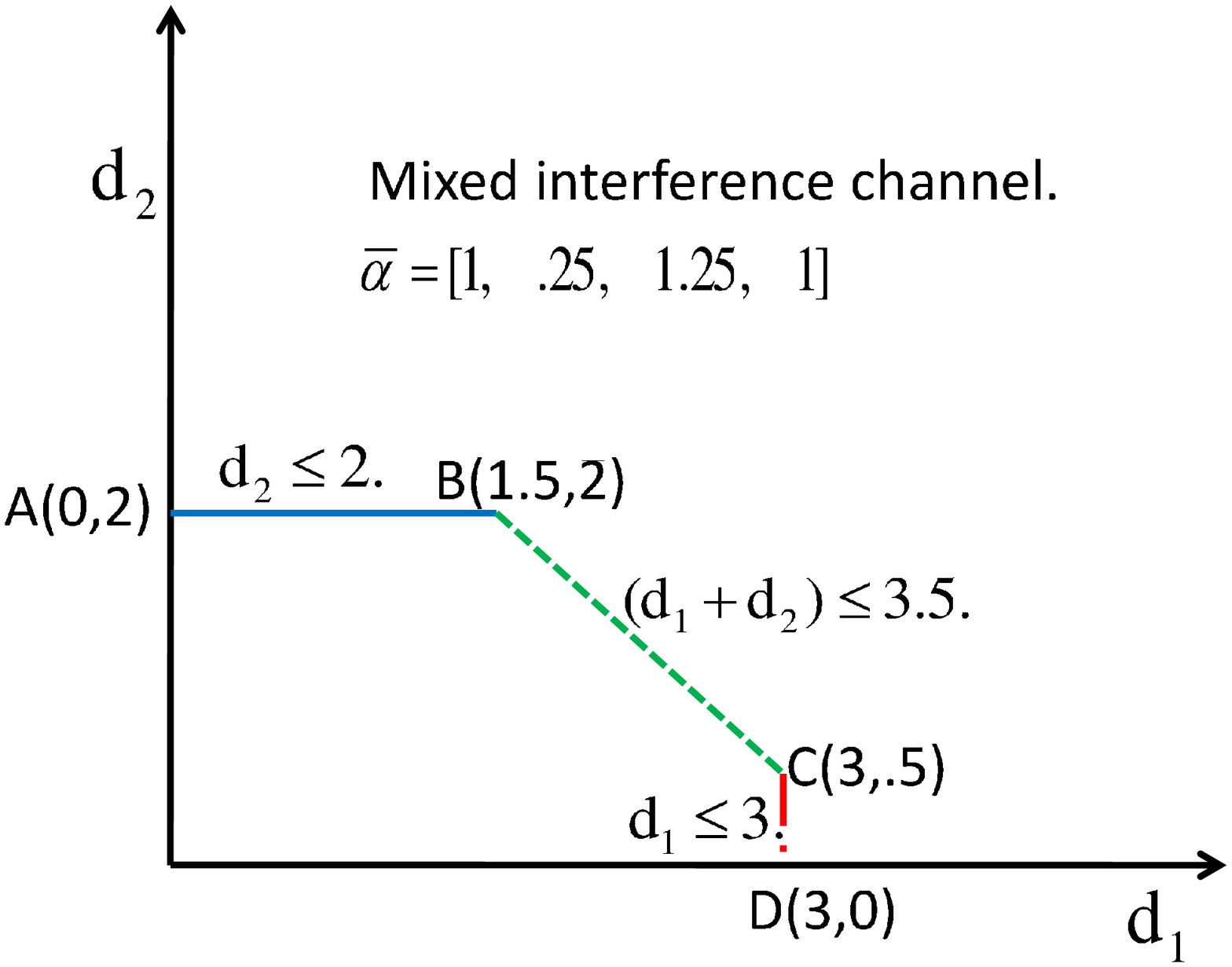}}
  \end{center}
\caption{GDoF region of the $(3,3,2,2)$ IC. }
\label{fig_explicit_scheme}
\end{figure}

\begin{ex}[A MIMO IC with weak interference]
\label{ex:weak-3322-interference channel}
Figure~\ref{fig_gdof_achievable_1} depicts the GDoF region of a $(3,3,2,2)$ MIMO IC with $\bar{\alpha}=[1,\frac{3}{5}, \frac{3}{5},1]$. Clearly, it is sufficient to illustrate the achievability of the vertices of the GDoF region since any point on the line joining any two vertices can be achieved via time-sharing. The time sharing argument, however, is just a matter of convenience, and is not necessary to achieve a point in the GDoF region. In fact, we illustrate the achievability of a non-corner point in Remark~\ref{rem_time_sharing}. Note that points A or E can be achieved simply by turning off $Tx_1$ or $Tx_2$, respectively. To analyze the achievability of the other corner points we need to know the DoFs carried by the private and public messages of each user. For the $(3,3,2,2)$ IC with $\bar{\alpha}=[1,\frac{3}{5}, \frac{3}{5},1]$, Lemma~\ref{lem:DoF-split} gives

\begin{equation}
\label{eq:dof-bounds-for-example-3322-a}
\begin{array}{cc}
d_{1p}\leq & 1.8;\\
d_{1c}\leq &2;\\
d_{2c}\leq &1.2;\\
(d_{1p}+d_{1c})\leq & 3;\\
(d_{1p}+d_{2c})\leq & 2.2;\\
(d_{1c}+d_{2c})\leq & 2.6;\\
(d_{1p}+d_{1c}+d_{2c})\leq & 3;
\end{array}~\textrm{and}~
\begin{array}{cc}
d_{2p}\leq & .8;\\
d_{2c}\leq &2;\\
d_{1c}\leq &1.2;\\
(d_{2p}+d_{2c})\leq & 2;\\
(d_{2p}+d_{1c})\leq & 1.2;\\
(d_{2c}+d_{1c})\leq & 2;\\
(d_{2p}+d_{2c}+d_{1c})\leq & 2;
\end{array}
\end{equation}

{\it Achievability of point B}: From the set of bounds in equation \eqref{eq:dof-bounds-for-example-3322-a} we see the only choice for the different DoFs for the public and private messages of the two users are given as $d_{1p}=1$, $d_{1c}=0$, $d_{2p}=.8$ and $d_{2c}=1.2$. Since the first user needs to send only private information having DoF 1, it is best to send it in the direction of the null space of $H_{12}$, i.e.,
\begin{IEEEeqnarray}{l}
X_1=\frac{1}{\sqrt{3}}x_{1p}^{(3)}U_{12}^{[3]}.
\end{IEEEeqnarray}
On the other hand, the structure of the codeword for the second user is also clear from equation \eqref{eq_structure_of_streams},
\begin{IEEEeqnarray}{rl}
X_2=&\sum_{k=1}^{2}\frac{\sqrt{\rho_{21} \lambda_{21}^{(k)}}}{\sqrt{2 (1+\rho_{21} \lambda_{21}^{(k)})}}x_{2c}^{(k)}U_{21}^{[k]}+\sum_{l=1}^{2}\frac{1}{\sqrt{2 (1+\rho_{21} \lambda_{21}^{(l)})}}x_{2p}^{(l)}U_{21}^{[l]},
\end{IEEEeqnarray}
where $x_{2c}^{(k)}$ and $x_{2p}^{(k)}$ carries $.6$ and $.4$ DoFs, respectively for both $k=1,2$.

{\it Decoding}: $Rx_1$ first projects the received signal on the 2 dimensional space which is perpendicular to $H_{11}U_{12}^{[3]}$ to remove the effect of $x_{1p}^{(3)}$ by zero forcing. In the resulting 2 dimensional signal space, only contribution from $W_2$ is present, carrying a DoF of $1.2$. This can be decoded because the link from $Tx_2$ to $Rx_1$ is a $2\times 2$ point-to-point MIMO channel with effective SNR of $\rho^{.6}$. Once decoded, $Rx_1$ removes its effect from the original received signal (the received signal before zero-forcing) and then it gets a interference-free channel from $Tx_1$ to itself. It can hence decode $U_1$. On the other hand, $Rx_2$ does not face any interference\footnote{The interference that reach below noise floor is irrelevant in the GDoF computation.} from $Tx_1$ so that it can decode $W_2$ while treating $U_2$ as noise. This is possible because treating $U_2$ as noise only raises the noise floor to $\rho^{.4}$ while the received signal power of $W_2$ is at $\rho$ which implies it can decode $.6$ DoFs from each receive dimension. Next, subtracting the contribution of $W_2$ from the received signal, $Rx_2$ can decode $U_2$.

{\it Achievability of point C}: Since $Rx_2$ can support only 2 DoFs at point C, we have $d_{1c}\leq .4$. Combining this with equation \eqref{eq:dof-bounds-for-example-3322-a} we get $d_{1c}= .4$, $d_{1p}= 1.4$, $d_{2p}= .8$ and $d_{2c}= .8$. For this choice of the different rates, the transmit signals at $Tx_1$ and $Tx_2$ are given by
\begin{IEEEeqnarray}{rl}
X_1=&\sum_{k=1}^{2}\frac{\sqrt{\rho_{12} \lambda_{12}^{(k)}}}{\sqrt{3 (1+\rho_{12} \lambda_{12}^{(k)})}}x_{1c}^{(k)}U_{12}^{[k]}+\sum_{l=1}^{2}\frac{1}{\sqrt{3 (1+\rho_{12} \lambda_{12}^{(l)})}}x_{1p}^{(l)}U_{12}^{[l]}+\frac{1}{\sqrt{3}}x_{1p}^{(3)}U_{12}^{[3]},
\end{IEEEeqnarray}
where $x_{1c}^{(k)}$ and $x_{1p}^{(k)}$ carries $.2$ DoFs, for both $k=1,2$ and $x_{1p}^{(3)}$ carries 1 DoF and
\begin{IEEEeqnarray}{rl}
X_2=&\sum_{k=1}^{2}\frac{\sqrt{\rho_{21} \lambda_{21}^{(k)}}}{\sqrt{2 (1+\rho_{21} \lambda_{21}^{(k)})}}x_{2c}^{(k)}U_{21}^{[k]}+\sum_{l=1}^{2}\frac{1}{\sqrt{2 (1+\rho_{21} \lambda_{21}^{(l)})}}x_{2p}^{(l)}U_{21}^{[l]},
\end{IEEEeqnarray}
where $x_{2c}^{(k)}$ and $x_{2p}^{(k)}$ carries $.4$ DoFs each, for both $k=1,2$. The different signals at both the receivers are depicted in Fig. \ref{fig_achievability_1C}, where each stream is represented by a box the top level of which marks its signal strength and the vertical height is proportional to the DoFs carried by it. Note that, $x_{1p}^{(1)}$ though transmitted at a power level of $1$, does not appear at $Rx_2$ since it is transmitted along the null space of the channel from $Tx_1$ to $Rx_2$.
\begin{figure}[htp]
  \begin{center}
    \subfigure[Receive signal spaces at DoF pair $(1.8,1.6)$.]{\label{fig_achievability_1C}\includegraphics[scale=.3]{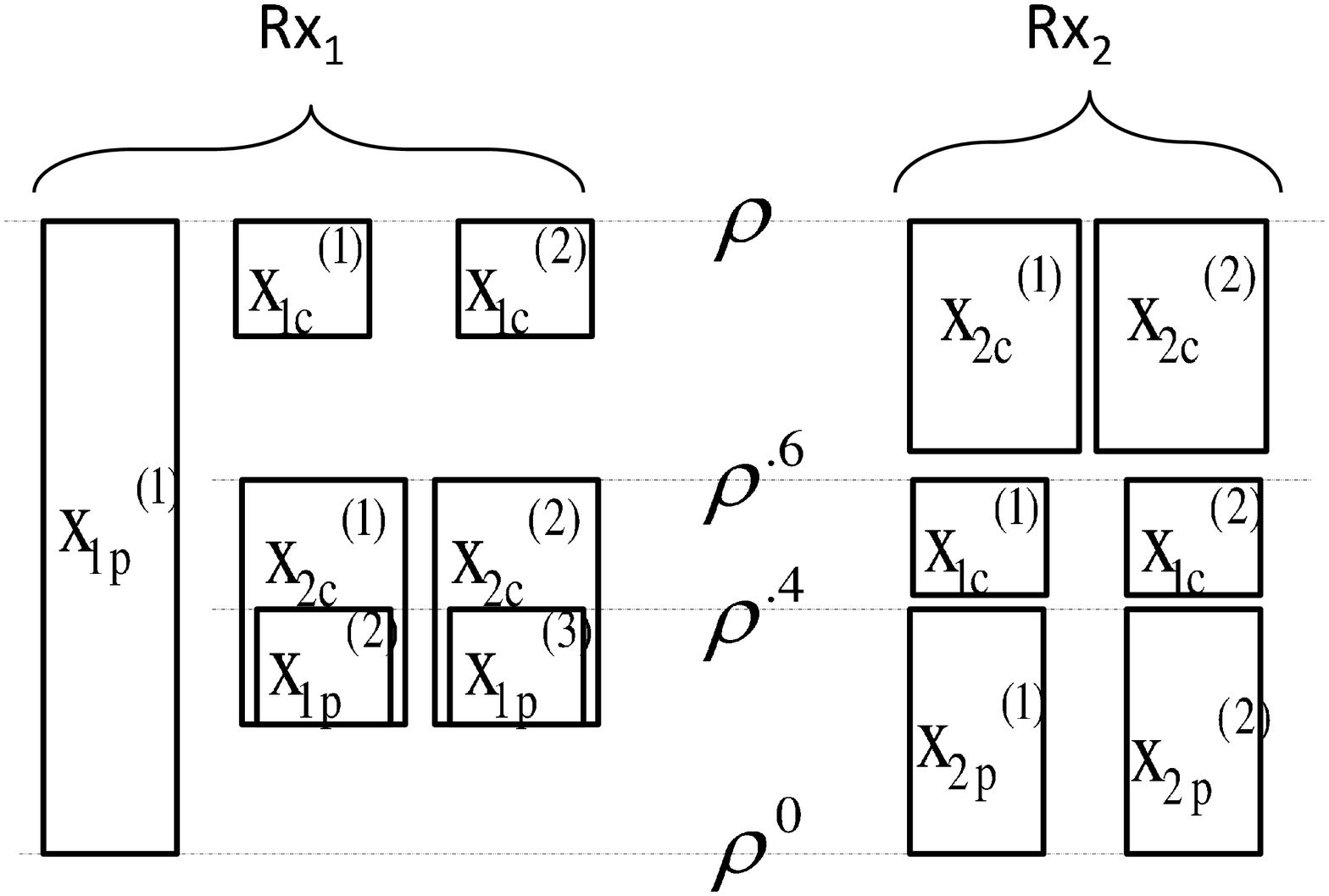}}
    \subfigure[Receive signal spaces at DoF pair $(2.6,.8)$.]{\label{fig_achievability_1D}\includegraphics[scale=.3]{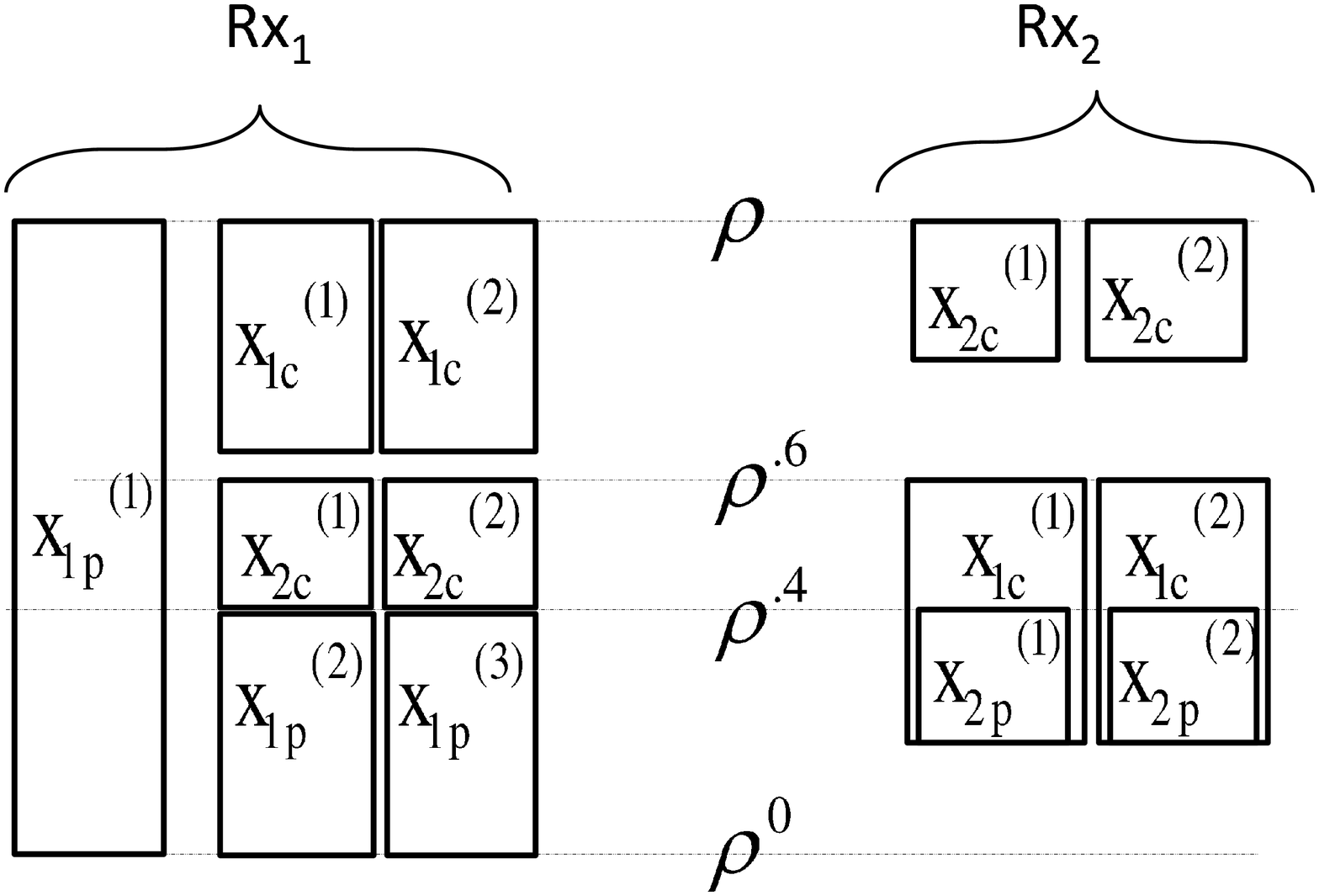}}
  \end{center}
\caption{GDoF region of a $(3,3,2,2)$ MIMO IC and its explicit achievable scheme.}
\label{fig_achievability}
\end{figure}

{\it Decoding}: The decoding procedure at $Rx_1$ is exactly the same as in the previous case. $Rx_2$ on the other hand, can decode $W_2$, $W_1$ and $U_2$, respectively, in that order through successive interference cancellation, i.e., it first decodes $W_2$, treating $W_1$ and $U_2$ (both of which are received below $\rho^{.6}$) as noise. Subtracting the contribution of $W_2$, it next decodes $W_1$ treating $U_2$ as noise. Finally, subtracting the contribution of $W_1$ it decodes $U_2$. It should be noted that during the decoding of each of these messages the noise floor is actually at the power level of the messages being treated as noise.

{\it Achievability of point D}: Again from \eqref{eq:dof-bounds-for-example-3322-a} we get $d_{1p}=1.8$, $d_{1c}=.8$, $d_{2p}=.4$ and $d_{2c}=.4$, for which $X_i$ for $1\leq i\leq 2$ can be written as
\begin{IEEEeqnarray}{rl}
X_1=&\sum_{k=1}^{2}\frac{\sqrt{\rho_{12} \lambda_{12}^{(k)}}}{\sqrt{3 (1+\rho_{12} \lambda_{12}^{(k)})}}x_{1c}^{(k)}U_{12}^{[k]}+\sum_{l=1}^{2}\frac{1}{\sqrt{3 (1+\rho_{12} \lambda_{12}^{(l)})}}x_{1p}^{(l)}U_{12}^{[l]}+\frac{1}{\sqrt{3}}x_{1p}^{(3)}U_{12}^{[3]},
\end{IEEEeqnarray}
where $x_{1c}^{(k)}$ and $x_{1p}^{(k)}$ carries $.4$ DoFs each, for both $k=1,2$ and $x_{1p}^{(3)}$ carries 1 DoF and
\begin{IEEEeqnarray}{rl}
X_2=&\sum_{k=1}^{2}\frac{\sqrt{\rho_{21} \lambda_{21}^{(k)}}}{\sqrt{2 (1+\rho_{21} \lambda_{21}^{(k)})}}x_{2c}^{(k)}U_{21}^{[k]}+\sum_{l=1}^{2}\frac{1}{\sqrt{2 (1+\rho_{21} \lambda_{21}^{(l)})}}x_{2p}^{(l)}U_{21}^{[l]},
\end{IEEEeqnarray}
where $x_{2c}^{(k)}$ and $x_{2p}^{(k)}$ carries $.2$ DoFs each, for both $k=1,2$. The different received signals at both the receivers are depicted in Fig. \ref{fig_achievability_1D}. It is clear from Fig. \ref{fig_achievability_1D} that a MAC receiver can decode all the messages.
\end{ex}
\vspace{2mm}

\begin{ex}[A MIMO IC with mixed interference]
\label{ex:mixed-3322-ic}
Fig. \ref{fig_gdof_achievable_2} depicts the GDoF region of the $(3,3,2,2)$ IC with {\it mixed} interference, i.e., $\alpha_{12}=.25< 1$ and $\alpha_{21}=1.25>1$. Points A and D of this channel are achievable simply by making $Tx_1$ and $Tx_2$, respectively, silent. Now for the MIMO IC of Fig. \ref{fig_gdof_achievable_2}, from Lemma \ref{lem:DoF-split}, we get

\begin{equation}
\label{eq:dof-bounds-for-example-3322-b}
\begin{array}{cc}
d_{1p}\leq & 2.5;\\
d_{1c}\leq &2;\\
d_{2c}\leq &2.5;\\
(d_{1p}+d_{1c})\leq & 3;\\
(d_{1p}+d_{2c})\leq & 3.5;\\
(d_{1c}+d_{2c})\leq & 3.5;\\
(d_{1p}+d_{1c}+d_{2c})\leq & 3.5;
\end{array}~\textrm{and}~
\begin{array}{cc}
d_{2p}\leq & 0;\\
d_{2c}\leq &2;\\
d_{1c}\leq &.5;\\
(d_{2p}+d_{2c})\leq & 2;\\
(d_{2p}+d_{1c})\leq & .5;\\
(d_{2c}+d_{1c})\leq & 2;\\
(d_{2p}+d_{2c}+d_{1c})\leq & 2;
\end{array}
\end{equation}

{\it Achievability of point B}: From equation \eqref{eq:dof-bounds-for-example-3322-b} we get $d_1=d_{1p}=1.5$ and $d_{2}=d_{2c}=2$, which imply $d_{1c}=d_{2p}=0$. With this choice of the GDoFs carried by the private and public messages we have
\begin{IEEEeqnarray}{l}
X_1=\sum_{k=1}^{2}\frac{1}{\sqrt{3(1+\rho_{12} \lambda_{12}^{(k)})}}x_{1p}^{(k)}U_{12}^{[k]}+\frac{1}{\sqrt{3}}x_{1p}^{(3)}U_{12}^{[3]},
\end{IEEEeqnarray}
where $x_{1p}^{(k)}$ carry $.25$, $.25$ and $1$ DoFs, for $k=1,2$ and $3$, respectively. On the other hand,
\begin{IEEEeqnarray}{rl}
X_2=&\sum_{k=1}^{2}\frac{1}{\sqrt{2}}x_{2c}^{(k)}{U_{21}}^{[k]},
\end{IEEEeqnarray}
where $x_{2c}^{(k)}$ carries $1$ DoFs for both $k=1,2$.

{\it Decoding}: $Rx_1$ first projects the received signal on the 2-dimensional space which is perpendicular to $H_{11}U_{12}^{[3]}$ to remove the effect of $x_{1p}^{(3)}$ by zero-forcing. In the resulting 2-dimensional signal space, contribution from $U_1$ and $W_2$ are present, together they are carrying $2.5$ DoFs and can be decoded by $Rx_1$, since as a MAC receiver it has a sum GDoF of $2.5$. Once decoded, it removes the effect of these signals from the original received signal (the received signal before zero-forcing), and decodes $x_{1p}^{(3)}$. On the other, $Rx_2$ does not face any interference from $Tx_1$ and so it can decode $W_2$.

{\it Achievability of point C}: Similarly, at point C we have $d_{1c}=.5$, $d_{1p}=2.5$, $d_{2c}=.5$, $d_{2p}=0$ and the transmit signals are
\begin{IEEEeqnarray}{rl}
X_1=&\sum_{k=1}^{2}\frac{\sqrt{\rho_{12} \lambda_{12}^{(k)}}}{\sqrt{3 (1+\rho_{12} \lambda_{12}^{(k)})}}x_{1c}^{(k)}U_{12}^{[k]}+\sum_{l=1}^{2}\frac{1}{\sqrt{3 (1+\rho_{12} \lambda_{12}^{(l)})}}x_{1p}^{(l)}U_{12}^{[l]}+\frac{1}{\sqrt{3}}x_{1p}^{(3)}U_{12}^{[3]},
\end{IEEEeqnarray}
where $x_{1c}^{(k)}$ and $x_{1p}^{(k)}$ carry $.25$ and $.75$ DoFs, respectively, for both $k=1,2$ and $x_{1p}^{(3)}$ carries 1 DoF. On the other hand,
\begin{IEEEeqnarray}{rl}
X_2=&\sum_{k=1}^{2}\frac{1}{\sqrt{2}}x_{2c}^{(k)}U_{21}^{[k]},
\end{IEEEeqnarray}
where $x_{2c}^{(k)}$ carries $.25$ DoFs, for both $k=1,2$.

{\it Decoding}: At $Rx_1$, first the contribution of $x_{1p}^{(3)}$ is zero-forced, and then $W_2$, $W_1$ and the remaining part of $U_1$ are decoded successively, in that order. During the decoding of each of these messages the others are treated as noise. After decoding a message, it is subtracted from the original signal. At $Rx_2$ also $W_2$, $W_1$ are decoded successively following a similar method.

\begin{rem}[ Achievability of a non-corner point]
\label{rem_time_sharing}
To illustrate the fact that all the points in between the corner points can be achieved without time-sharing, we consider the point $(2.25,1.25)$. From equation \eqref{eq:dof-bounds-for-example-3322-b} we know that $d_{1p}=1.75, d_{1c}=.5$ and $d_2=d_{2c}=1.25$ is achievable. The corresponding codewords are
\begin{IEEEeqnarray}{rl}
X_1=&\sum_{k=1}^{2}\frac{\sqrt{\rho_{12} \lambda_{12}^{(k)}}}{\sqrt{3 (1+\rho_{12} \lambda_{12}^{(k)})}}x_{1c}^{(k)}U_{12}^{[k]}+\sum_{l=1}^{2}\frac{1}{\sqrt{3 (1+\rho_{12} \lambda_{12}^{(l)})}}x_{1p}^{(l)}U_{12}^{[l]}+\frac{1}{\sqrt{3}}x_{1p}^{(3)}U_{12}^{[3]},
\end{IEEEeqnarray}
where $x_{1c}^{(k)}$ and $x_{1p}^{(k)}$ carry $.25$ and $3/8$ DoFs, respectively for both $k=1,2$ and $x_{1p}^{(3)}$ carries 1 DoF. On the other hand,
\begin{IEEEeqnarray}{rl}
X_2=&\sum_{k=1}^{2}\frac{1}{\sqrt{2}}x_{2c}^{(k)}U_{21}^{[k]},
\end{IEEEeqnarray}
where $x_{2c}^{(k)}$ carries $5/8$ DoFs, for both $k=1,2$.

{\it Decoding}: Decoding is similar to methods described above. The details are skipped.
\end{rem}
\end{ex}


\begin{figure}[htp]
  \begin{center}
    \subfigure[$\alpha=\frac{2}{3}$]{\label{figure_gdof_3232-a}\includegraphics[scale=0.3]{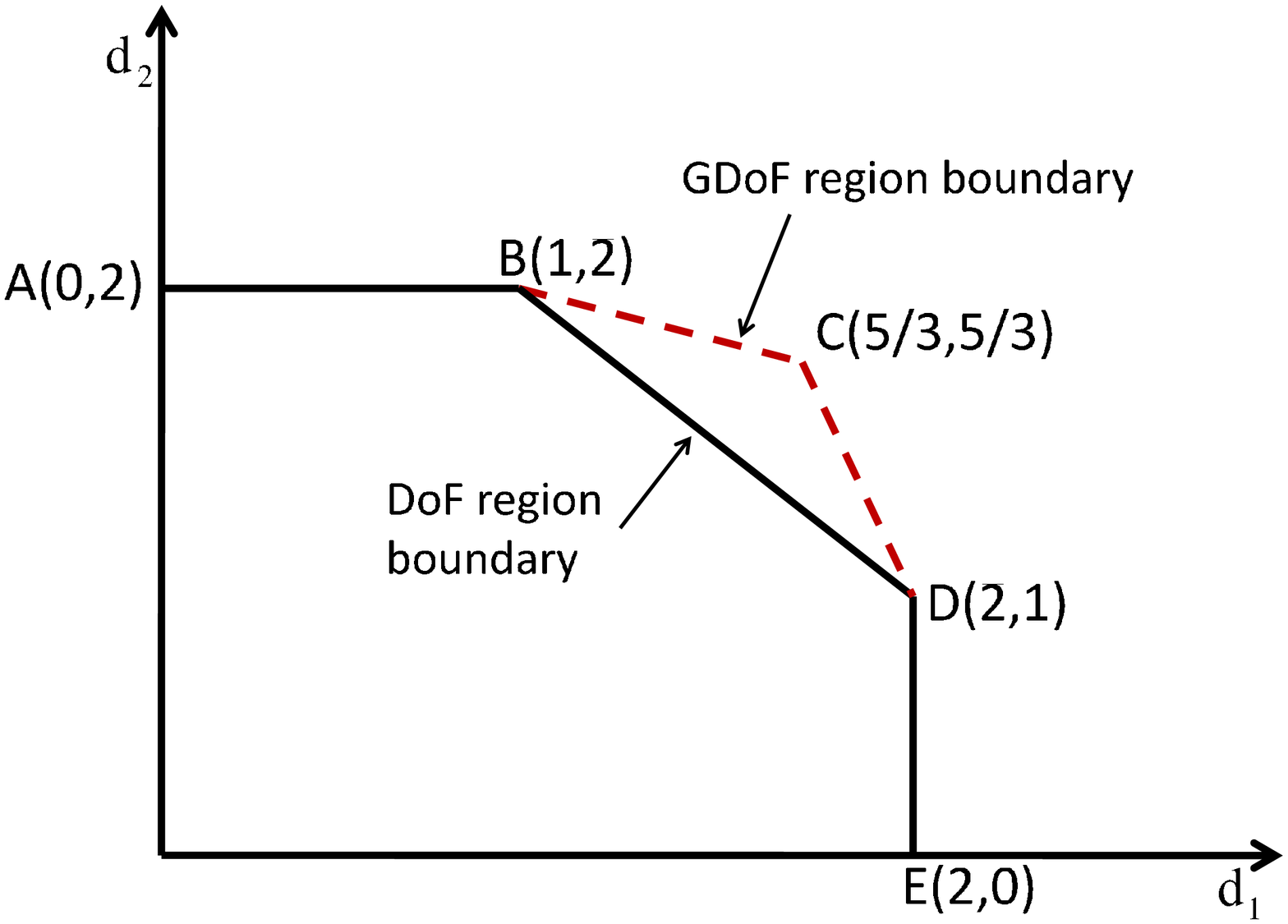}}
    \subfigure[$\alpha=\frac{3}{2}$]{\label{figure_gdof_3232-b}\includegraphics[scale=0.3]{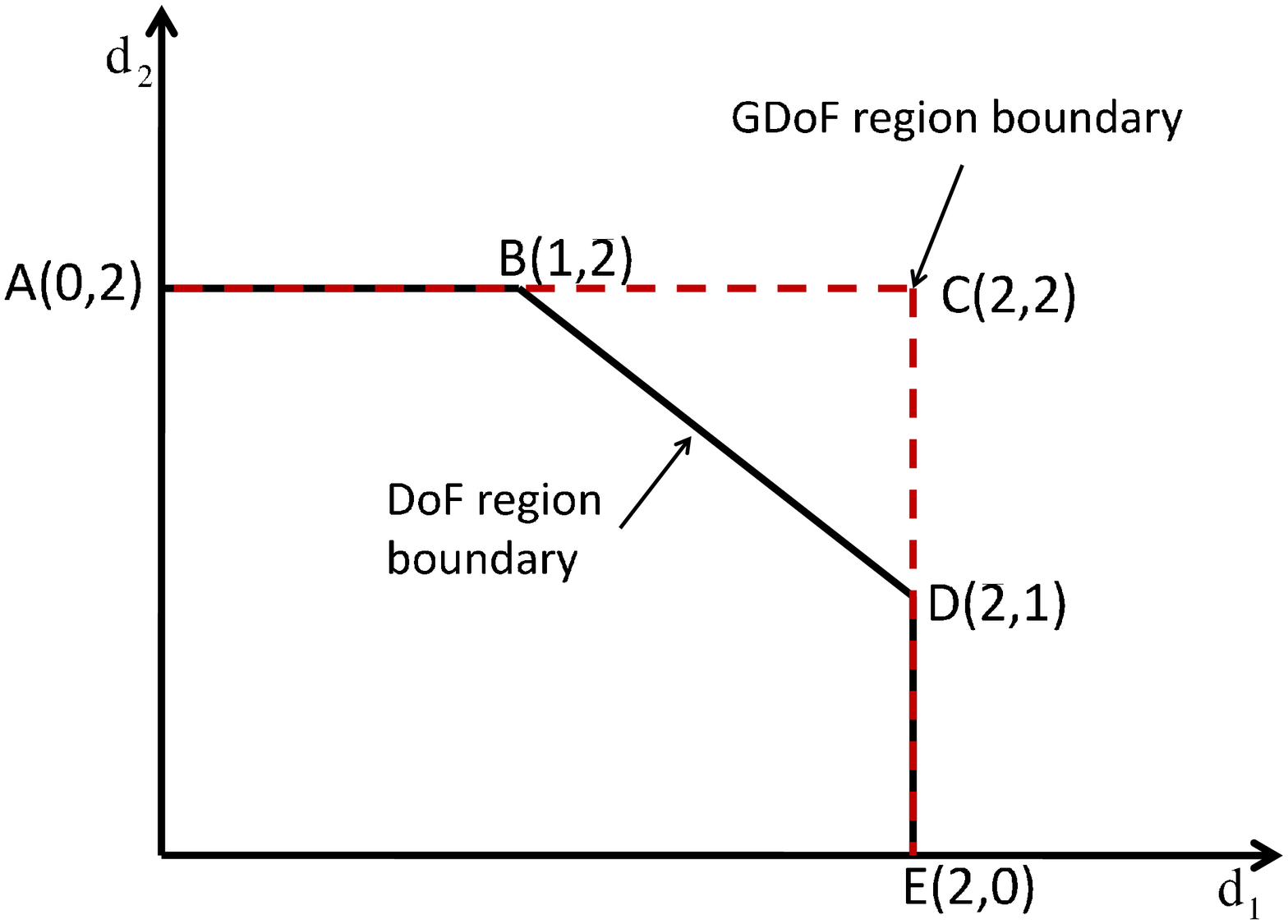}}
  \end{center}
\caption{GDoF region of the $(3,2,3,2)$ MIMO IC with $\alpha_{11}=\alpha_{22}=1$ and $\alpha_{12}=\alpha_{21}=\alpha$. }
\label{figure_gdof_3232}
\end{figure}

\section{Further insights}
\label{sec_insights}
\vspace{-1.25mm}
\subsection{Only Tx/Rx ZF Beam-forming is not GDoF optimal}
\label{subsec_BFZF_suboptimality}
The fundamental GDoF gives a finer high SNR approximation than the DoF approximation and therefore reveals insights that are not revealed by the DoF analysis. Figure~\ref{figure_gdof_3232-a} illustrates this point by comparing the DoF and GDoF region of the $(3,2,3,2)$ IC with $\bar{\alpha}=[1,\frac{2}{3},\frac{2}{3},1]$. It is known from \cite{JFak} that only transmit/receive zero-forcing beam-forming is sufficient to achieve any point in the DoF region of the channel. The DoF region achievable using this scheme is shown in Fig. \ref{figure_gdof_3232} as against the fundamental GDoF region. It is easily seen that forgoing the opportunity to align signals in the signal-level dimension leads to a strictly GDoF suboptimal performance. In particular, this technique can not achieve any point in the triangular region BCD. However, the coding scheme of Section~\ref{sec_channel_model_and_preliminaries} which in addition to beamforming, also employs signal-level interference alignment, can achieve all the points in the region BCD.

\subsection{On achieving single-user performance}
\label{subsection_single_user_perf}
It is well known that the achievability of single user DoFs on a MIMO MAC or BC depends on the number of antennas at the different nodes. Moreover, on a SISO IC, it depends on the interference level, $\alpha$. On a MIMO IC it depends on {\em both}. From Corollary~\ref{cor_gdof_symmetric_MgeqN} we get that on a $(M,N,M,N)$ IC with $M\geq N$, the single user GDoF is achieved (each user gets $N$ DoFs) when
\begin{equation*}
    \alpha\geq \alpha^*=\left(3-\frac{M}{N}\right).
\end{equation*}
In contrast to the case on a SISO IC, the value of $\alpha$ at which the single-user performance is achieved, denoted by $\alpha^*$, decreases below $2$ as $M$ increases, giving a similar effect as in a MAC or BC (see Fig. \ref{figure_gdof_MgeqN_channel}).

\begin{figure}[htp]
  \begin{center}
    \subfigure[Symmetric GDoF region of the $(3,2,2,3)$ IC.]{\label{figure_gdof_3223_channel}\includegraphics[scale=.5]{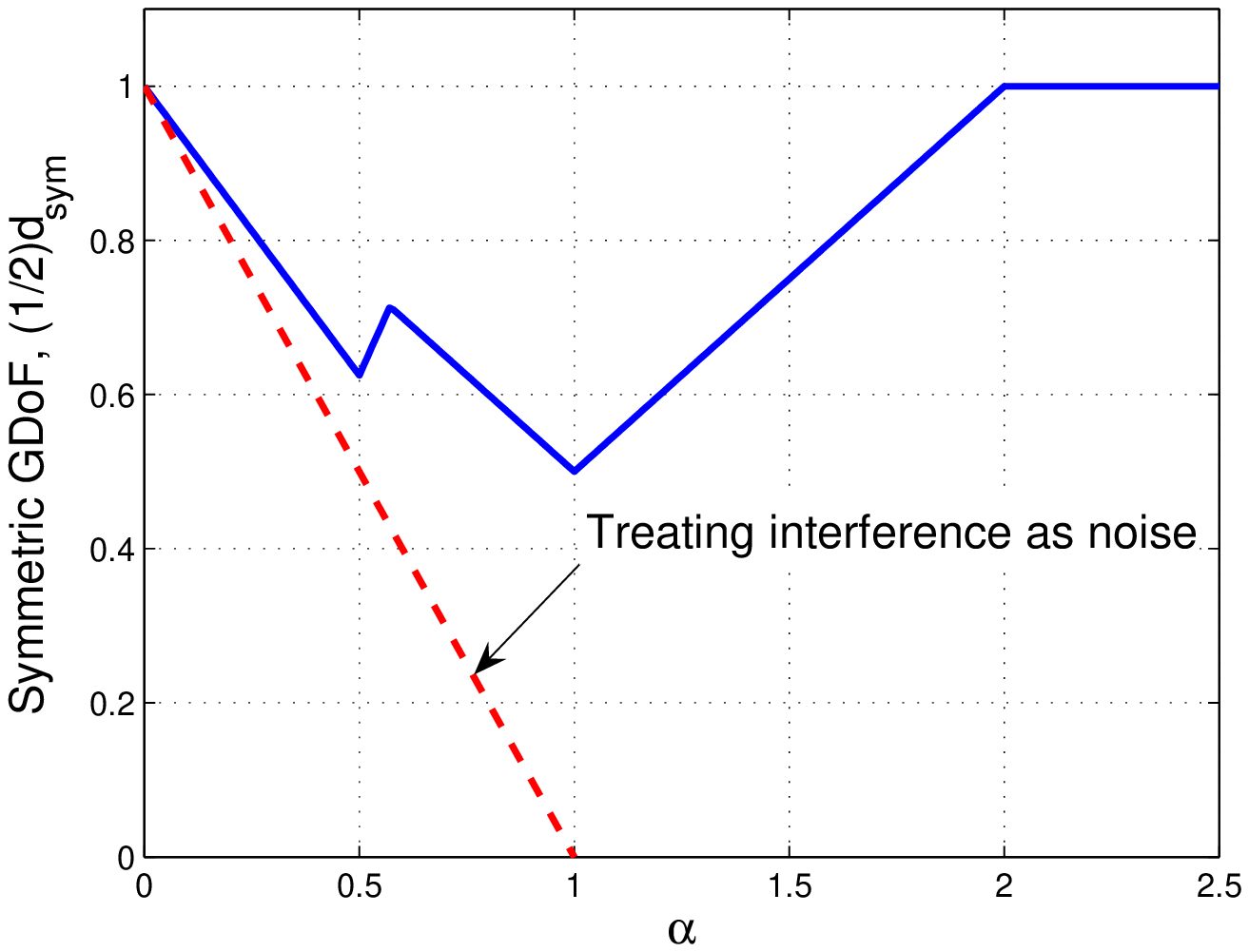}}
    \subfigure[Symmetric GDoF region of a $(1,1,2,1)$ MIMO IC]{\label{figure_gdof_1121}\includegraphics[scale=0.5]{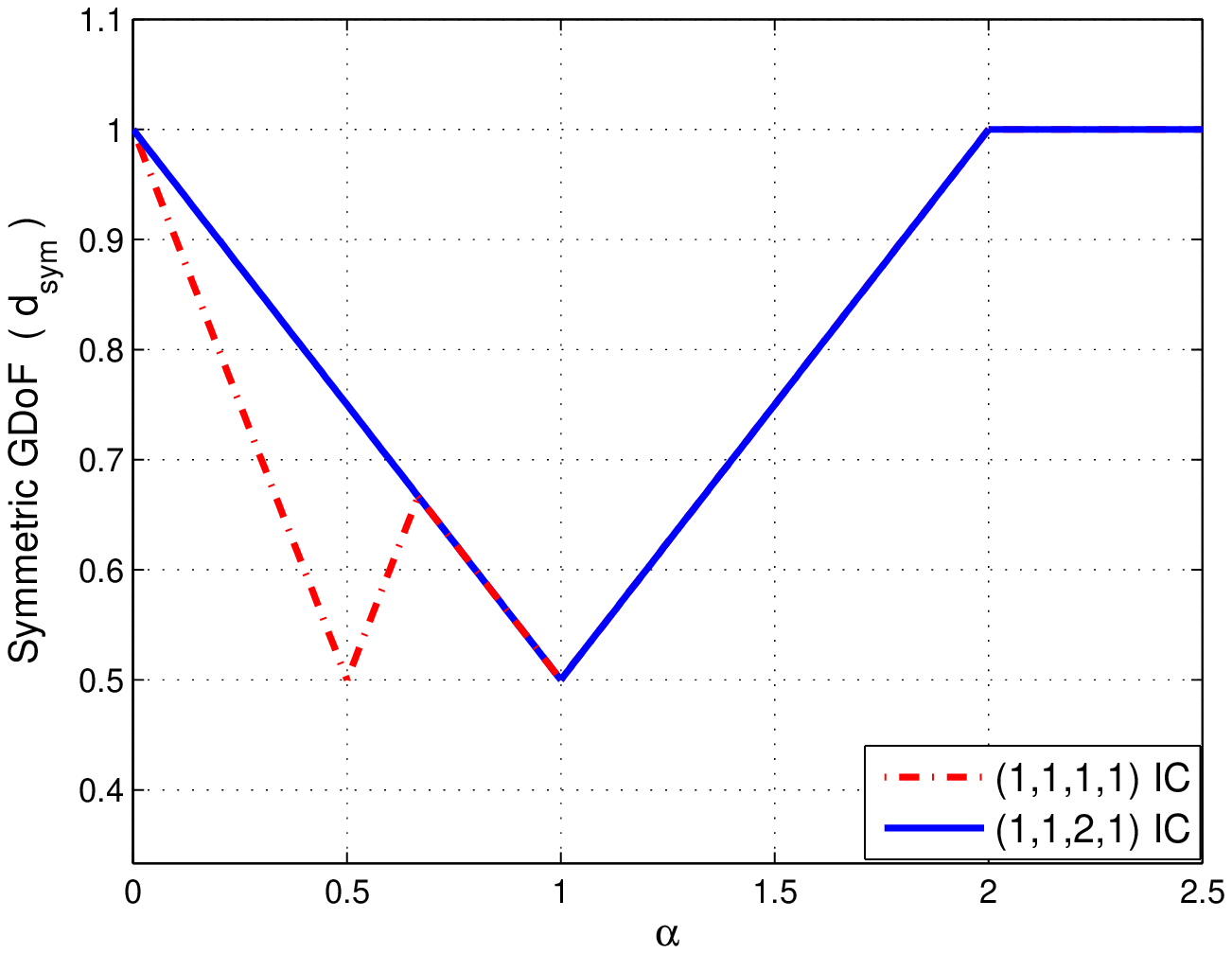}}
  \end{center}
\caption{Sub-optimality of TIN and deviation of the GDoF boundary from the well known ``W" shape. }
\end{figure}

\subsection{Sub-optimality of treating interference as noise}
Another fundamental difference of the MIMO IC from the SISO IC revealed by the GDoF analysis is this: in general, treating interference as noise (TIN) is {\em not} GDoF optimal on a MIMO IC even in the {\it very weak interference} regime, i.e., when $\alpha\leq \frac{1}{2}$. This is seen in Fig. \ref{figure_gdof_MgeqN_channel} where the dotted line, which represents the symmetric GDoF achievable by TIN, is strictly sub-optimal with respect to the fundamental GDoF of the channel for $\alpha\leq \frac{1}{2}$ whenever $M/N>1$. See also Fig. \ref{figure_gdof_3223_channel} which illustrates this point for the $(3,2,2,3)$ MIMO IC.

\vspace{-1.25mm}

\begin{figure}[htp]
  \begin{center}
  \subfigure{\label{figure_ZF_channel-b}\includegraphics[scale=.25]{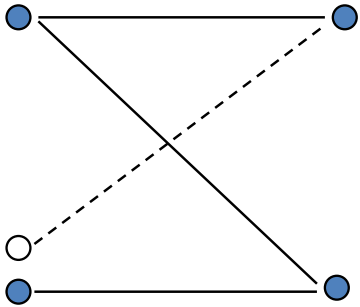}}
  \end{center}
\caption{Diagonalization of the cross links using ZF and BF. }
\end{figure}

\subsection{Deviation from the ``W" shape}
Unlike in the SISO IC, the symmetric GDoF region of a MIMO IC in general need not maintain the ``W" shape. The deviation in general is due to asymmetry in the numbers of antennas. For example, consider the $(1,1,2,1)$ IC with $\alpha_{ii}=1$ and $\alpha_{ij}=\alpha$, for $i\neq j\in \{1,2\}$. The best achievable symmetric DoF ($d_{sym}=d_1=d_2$) on this channel denoted by $d$ is
\begin{equation*}
d=\left\{\begin{array}{ll}
1-\frac{\alpha}{2}, &0\leq \alpha\leq 1;\\
\frac{\alpha}{2}, & 1\leq\alpha\leq 2;\\
1, & 2\leq \alpha.
\end{array}\right.
\end{equation*}
which is depicted in Figure~\ref{figure_gdof_1121}. Diagonalizing the cross-link from $Tx_2$ to $Rx_1$ and then turning off the subchannel which interferes with $Rx_1$ gives the GDoF equivalent channel of Figure~\ref{figure_ZF_channel-b} which is a SISO ``Z" IC. The symmetric GDoF region of this channel is indeed ``V" shaped as found in \cite{ETW1}. Although a little more involved, the distorted ``W" of Figure~\ref{figure_gdof_3223_channel} for the $(3,2,2,3)$ MIMO IC can be explained similarly.

\section{Conclusion}
\label{sec_conclusion}

The GDoF analysis of this paper, unifies and generalizes the earlier results on GDoF of SISO IC~\cite{ETW1}, the DoF region~\cite{JFak} of MIMO IC and the symmetric GDoF~\cite{PBT} of MIMO IC through a single achievable scheme for all. The coding schemes in~\cite{JFak} and \cite{PBT} are strictly suboptimal in the GDoF sense on a general 2-user MIMO IC in one case or other. The analysis here reveals various insights about the MIMO IC including the fact that in general, partially decoding the unintended user's message is necessary to be GDoF optimal even in the so called very weak interference regime. The two types of signaling dimensions available on a MIMO IC -- namely, signal space and signal level -- are jointly and optimally exploited in the GDoF optimal scheme.

\appendices


\section{Proof of Lemma~\ref{lem:sum-dof-of-2user-MAC}}
\label{App:proof-lemma-2user-MAC}
We shall consider two different cases: 1) $(u_1+u_2)\geq u$; and 2) $(u_1+u_2)< u$. The first case was proved in Lemma~1 of \cite{PBT} which gives
\begin{equation}\label{eq_lem_my_MAC_1}
\log\det\left(I_u+\rho^a H_1H_1^{\dagger}+\rho^b H_2H_2^{\dagger}\right)=\min\{u,u_1\}a\log(\rho)+(u-u_1)^+b\log(\rho)+\mathcal{O}(1).
\end{equation}
However when $(u_1+u_2)< u$ we have the following
\begin{IEEEeqnarray}{rl}
\log\det\left(I_u+\rho^a H_1H_1^{\dagger}+\rho^b H_2H_2^{\dagger}\right)=&\log\det\left(I_u+ [H_1~H_2]\left[\begin{array}{cc}\rho^aI_{u_1} & 0\\ 0 & \rho^aI_{u_2} \end{array}\right] [H_1~H_2]^{\dagger}\right)\nonumber\\
=&\log\det\left(I_{(u_1+u_2)}+ H^\dagger H\left[\begin{array}{cc}\rho^aI_{u_1} & 0\\ 0 & \rho^aI_{u_2} \end{array}\right] \right),~[\because H=[H_1~H_2]]\nonumber\\
\stackrel{(a)}{=}&\log\det\left( H^\dagger H\left[\begin{array}{cc}\rho^aI_{u_1} & 0\\ 0 & \rho^aI_{u_2} \end{array}\right] \right)+o(1)\nonumber\\
\label{eq_lem_my_MAC_2}
=&(u_1a+u_2b)\log(\rho)+\mathcal{O}(1),
\end{IEEEeqnarray}
where step $(a)$ follows from the fact that $H^\dagger H$ is full rank, since $H$ is full rank by assumption. Finally combining equations \eqref{eq_lem_my_MAC_1} and \eqref{eq_lem_my_MAC_2} we have the desired result.

\section{Proof of Lemma~\ref{lem:sum-dof-of-3user-MAC}}
\label{pf_lem_my_MAC3}
Without loss of generality, let us assume that $a_1\geq \max\{a_2,a_3\}$. In the proof we shall use the following notations:
\begin{IEEEeqnarray*}{rl}
L_t=\log\det\left(I_u+\sum_{i=1}^3\rho^{a_i}H_iH_i^{\dagger}\right),~\Lambda = \left[\begin{array}{cc} \rho^{a_2}I_{u_2} & 0\\ 0 & \rho^{a_3}I_{u_3}\end{array}\right],~\textrm{and}~ H_{23}=[H_2~H_3],
\end{IEEEeqnarray*}
where the entries of $H_{23}$ are iid that come form a continuous distribution and hence $H_{23}$ is is full rank w.p.1. Using the identity $\log\det(I+AB)=\log\det(I+BA)$ we get
\begin{IEEEeqnarray}{rl}
L_t=&\log\det\left(I_{u}+\rho^{a_1}H_1H_1^{\dagger}+H_{23}\Lambda H_{23}^{\dagger}\right),\nonumber \\
=&\log\det\left(I_{u}+\rho^{a_1}H_1H_1^{\dagger}\right)+\log\det\left(I_{u_2+u_3}+ \Lambda H_{23}^{\dagger}\left(I_u+\rho^{a_1}H_1H_1^{\dagger}\right)^{-1}
H_{23}\right),\nonumber\\
\label{eq_pf_lem_my_MAC3_1}
=&\min \{u,u_1\} a_1\log(\rho)+\log\det\left(I_{u_2+u_3}+ \Lambda H_{23}^{\dagger}\left(I_u+\rho^{a_1}H_1H_1^{\dagger}\right)^{-1}
H_{23}\right)+\mathcal{O}(1),
\end{IEEEeqnarray}
Next, we approximate the second term on the right hand side of the last equation as
\begin{IEEEeqnarray}{rl}
L_{t2}=&\log\det\left(I_{u_2+u_3}+ \Lambda H_{23}^{\dagger}\left(I_u+\rho^{a_1}H_1H_1^{\dagger}\right)^{-1}
H_{23}\right) \nonumber \\
=&\log\det\left(I_{u_2+u_3}+\Lambda H_{23}^\dagger U\left[\begin{array}{cc}\left(I_{\min\{u,u_1\}}+\rho^{a_1}\Lambda_1\right)^{-1}& 0\\ 0 & I_{(u-u_1)^+}\end{array}\right]U^{\dagger}H_{23}\right),
\end{IEEEeqnarray}
where in the last step we have used the eigen-decomposition of the matrix $H_1H_1^\dagger$, where $\Lambda_1$ is a diagonal matrix containing the positive eigenvalues only and $U$ is the unitary matrix containing all the eigen-vectors. Since both $H_2,H_3\in \mathcal{P}$, $H\triangleq U^\dagger H_{23}$ is identically distributed as $H_{23}$ and thus $H\in \mathcal{P}$. Suppose the rows of the matrix $H$ are divided into two sub sets: $G_1^\dagger=H^{(1:\min\{u,u_1\})}$ and $G_2^\dagger = H^{(\min\{u,u_1\}+1:u)}$, then both $G_1^\dagger, G_2^\dagger \in \mathcal{P}$ and from the last equation we get
\begin{IEEEeqnarray}{rl}
L_{t2}=&\log\det\left(I_{u_2+u_3}+\Lambda G_1\left(I_{u_1}+\rho^{a_1}\Lambda_1\right)^{-1}G_1^{\dagger}+\Lambda G_2G_2^\dagger\right), \nonumber \\
\stackrel{(c)}{=}&\log\det\left(I_{u_2+u_3}+\Lambda G_2 G_2^{\dagger}\right)+\mathcal{O}(1), \nonumber \\
\label{eq_pf_lem_my_MAC3_2}
=&\log\det\left(I_{(u-u_1)^+}+G_2^\dagger \Lambda G_2 \right)+\mathcal{O}(1),
\end{IEEEeqnarray}
where step $(c)$ follows from the fact that $a_1\geq \max\{a_2,a_3\}$. Since $G_2^\dagger\in \mathcal{P}$, it is full rank and so are $(G_2^\dagger)^{[1:u_2]}=G_{21}$ and $(G_2^\dagger)^{[u_2+1:u_2+u_3]}=G_{22}$. Putting these into equation \eqref{eq_pf_lem_my_MAC3_2} we get
\begin{IEEEeqnarray*}{rl}
L_{t2}=&\log\det\left(I_{u-u_1}+\rho^{a_2}G_{21}G_{21}^\dagger +\rho^{a_3} G_{22}G_{22}^\dagger\right)+\mathcal{O}(1),\nonumber\\
=& \left(\min\{(u-u_1)^+,u_2\}a_2+\min\{(u-u_1-u_2)^+,u_3\}a_3\right)\log(\rho)+\mathcal{O}(1),
\end{IEEEeqnarray*}
where the last step follows from Lemma~\ref{lem:sum-dof-of-2user-MAC}. Finally, substituting this into equation \eqref{eq_pf_lem_my_MAC3_1}, we get the desired result.


\section{Proof of Theorem~\ref{thm_mainresult}}
\label{proof_thm_mainresult}

{\it $1^{st}$ and $2^{nd}$ bound}: We start with the first two constraints in equation \eqref{eq_bound1} and \eqref{eq_bound2},
\begin{IEEEeqnarray*}{rl}
I_1\triangleq & \log\det\left(I_{N_1}+\rho H_{11}H_{11}^{\dagger}\right)=\min \{M_1, N_1\}\log(\rho)+\mathcal{O}(1);~\left[\because ~ \textrm{Lemma~\ref{lem:sum-dof-of-2user-MAC} with $a=1$ and $b=0$}\right]\\
I_2\triangleq & \log\det\left(I_{N_1}+\rho^{\alpha_{22}}H_{22}H_{22}^{\dagger}\right)=\min \{M_2, N_2\}\alpha_{22}\log(\rho)+\mathcal{O}(1).~\left[\because ~ \textrm{Lemma~\ref{lem:sum-dof-of-2user-MAC} with $a=\alpha_{22}$ and $b=0$}\right]
\end{IEEEeqnarray*}
Putting these into equations \eqref{eq_outline_gdof_b1} and \eqref{eq_outline_gdof_b2} we get
\begin{IEEEeqnarray}{l}
\label{eq_pf_D1}
d_1 \leq \lim_{\rho \to \infty} \frac{I_1}{\log(\rho)}=\min \{M_1, N_1\}; \\
\label{eq_pf_D2}
d_2 \leq \lim_{\rho \to \infty} \frac{I_2}{(\alpha_{22}\log(\rho))}=\min\{M_2, N_2\}.
\end{IEEEeqnarray}


{\it $3^{rd}$ and $4^{th}$ bound}: Using Lemma~\ref{lem:sum-dof-of-2user-MAC} we get
\begin{IEEEeqnarray}{rl}
I_{31}\triangleq &\log\det\left(I_{N_2}+\rho^{\alpha_{12}}H_{12}H_{12}^{\dagger}+\rho^{\alpha_{22}}H_{22}H_{22}^{\dagger}\right) \nonumber \\
\label{eq_pf_gdof_b31}
=& f\left(N_2, (\alpha_{12},M_1),(\alpha_{22},M_2)\right)\log(\rho)+\mathcal{O}(1),
\end{IEEEeqnarray}
where $f(.,.,.)$ is as defined in equation \eqref{eq_def_f}. The second term in equation \eqref{eq_bound3} can again be approximated by using Lemma~\ref{lem:sum-dof-of-2user-MAC} twice (recall Remark~\ref{rem:channel-distribution}) as follows.
\begin{IEEEeqnarray}{rl}
I_{32}=&\log \det \left(I_{N_1}+\rho_{11} H_{11}\left( I_{M_1}+\rho_{12} H_{12}^{\dagger}H_{12}\right)^{-1}H_{11}^{\dagger}\right),\nonumber \\
\label{eq_pf_mainthm_I32_expansion}
=&\log \det \left(I_{M_1}+\rho_{12} H_{12}^{\dagger}H_{12}+\rho_{11} H_{11}^\dagger H_{11}\right)-\log\det\left( I_{M_1}+\rho_{12} H_{12}^{\dagger}H_{12}\right),\nonumber \\
\label{eq_pf_gdof_b323}
=&f\left(M_1, (\alpha_{12},N_2),(\alpha_{11},N_1)\right)\log(\rho)-\min\{M_1,N_2\}\alpha_{12}\log(\rho)+\mathcal{O}(1),\\
=&f\left(N_1, ((1-\alpha_{12})^+,m_{12}),(1,(M_1-N_2)^+)\right)\log(\rho)+\mathcal{O}(1),
\end{IEEEeqnarray}
where the last step follows from the definition of $f(.)$, i.e., equation \eqref{eq_def_f}. Next, putting equations \eqref{eq_pf_gdof_b31}, and \eqref{eq_pf_gdof_b323} in equation \eqref{eq_outline_gdof_b3} we get
\begin{IEEEeqnarray*}{rl}
d_1 +\alpha_{22}d_2\leq &\lim_{\rho \to \infty} \frac{I_3}{\log(\rho)}=\lim_{\rho \to \infty} \frac{I_{31}+I_{32}}{\log(\rho)}; \nonumber \\
\label{eq_pf_D3}
=& f\left(N_2, (\alpha_{12},M_1),(\alpha_{22},M_2)\right)+f\left(N_1, ((1-\alpha_{12})^+,m_{12}),(1,(M_1-N_2)^+)\right).
\end{IEEEeqnarray*}
Similarly, approximating the terms in equation \eqref{eq_bound4} by Lemma~\ref{lem:sum-dof-of-2user-MAC} and putting it in equation \eqref{eq_outline_gdof_b4} we get
\begin{IEEEeqnarray*}{rl}
\label{eq_pf_D4}
d_1 +\alpha_{22}d_2\leq & f\left(N_1, (\alpha_{21},M_2),(\alpha_{11},M_1)\right)+f\left(N_2, ((\alpha_{22}-\alpha_{21})^+,m_{21}),(\alpha_{22},(M_2-N_1)^+)\right).
\end{IEEEeqnarray*}

{\it $5^{th}$ bound}: Note that neither of the terms in equation \eqref{eq_bound5} are in a form on which we can apply Lemma~\ref{lem:sum-dof-of-2user-MAC} or \ref{lem:sum-dof-of-3user-MAC}. However, as we shall see next, these terms can be expressed in an alternative format on which Lemma~\ref{lem:sum-dof-of-3user-MAC} can be used. Let the eigenvalue decomposition of the matrix $H_{12}^\dagger H_{12}$ is given as
\begin{IEEEeqnarray*}{l}
H_{12}^\dagger H_{12}=U\Lambda U^\dagger, ~\textrm{where} ~{\Lambda}=\left[\begin{array}{cc} \Lambda^+ & 0\\ 0 & 0_{(M_1-N_2)^+}\end{array}\right]
\end{IEEEeqnarray*}
and $\Lambda^+$ is a diagonal matrix containing only the positive eigenvalues. Using this decomposition we get
\begin{IEEEeqnarray*}{rl}
I_{51} \triangleq & \log\det\left(I_{N_1}+\rho^{\alpha_{21}}H_{21}H_{21}^{\dagger}+\rho H_{11}\left(I_{M_1}+\rho_{12}H_{12}^\dagger H_{12}\right)^{-1}H_{11}^{\dagger}\right)\nonumber \\
=&\log\det\left(I_{N_1}+\rho^{\alpha_{21}}H_{21}H_{21}^{\dagger}+\rho H_{11}U\left[\begin{array}{cc}\left(I_{m_{12}}+\rho^{\alpha_{12}}\Lambda^+\right)^{-1}& 0\\ 0 & I_{(M_1-N_2)^+}\end{array}\right]U^{\dagger}H_{11}^{\dagger}\right),\nonumber \\
=&\log\det\left(I_{N_1}+\rho^{\alpha_{21}}H_{21}H_{21}^{\dagger}+\rho \widetilde{H}\left[\begin{array}{cc}\left(I_{m_{12}}+\rho^{\alpha_{12}}\Lambda^+\right)^{-1}& 0\\ 0 & I_{(M_1-N_2)^+}\end{array}\right]\widetilde{H}^{\dagger}\right), ~\left[\because ~\widetilde{H}=H_{11}U\right]
\end{IEEEeqnarray*}
Note that $\widetilde{H}$ is identically distributed to $H_{11}$. Therefore, both $G_1=\widetilde{H}^{[1:m_{12}]}$ and $G_2=\widetilde{H}^{[(m_{12}+1):M_1]}$ have the same distribution as specified in section~\ref{sec_channel_model_and_preliminaries}, having all the properties of a typical channel matrix of the 2-user MIMO IC. Substituting this in the last equation we get
\begin{IEEEeqnarray*}{rl}
I_{51}=&\log\det\left(I_{N_1}+\rho^{\alpha_{21}}H_{21}H_{21}^{\dagger}+\rho G_1\left(I_{m_{12}}+\rho^{\alpha_{12}}\Lambda^+\right)^{-1}G_1^{\dagger}+\rho G_2G_2^{\dagger}\right) \nonumber \\
\label{eq_temp21}
=&\log\det\left(I_{N_1}+\rho^{\alpha_{21}}H_{21}H_{21}^{\dagger}+\rho^{(1-\alpha_{12})} G_1G_1^{\dagger}+\rho G_2G_2^{\dagger}\right)+o(1).
\end{IEEEeqnarray*}
Clearly, we can now apply Lemma~\ref{lem:sum-dof-of-3user-MAC} on
equation \eqref{eq_temp21},
\begin{IEEEeqnarray}{rl}
I_{51} =g\left(N_1,(\alpha_{21},M_2),((1-\alpha_{12})^+,m_{12}),(1,(M_1-N_2)^+)\right)\log(\rho)+\mathcal{O}(1).
\end{IEEEeqnarray}
Applying similar technique for the other term in equation \eqref{eq_bound5} we get
\begin{IEEEeqnarray}{rl}
I_{52} =g\left(N_2,(\alpha_{12},M_1),((\alpha_{22}-\alpha_{21})^+,m_{21}),(\alpha_{22},(M_2-N_1)^+)\right)\log(\rho)+\mathcal{O}(1).
\end{IEEEeqnarray}
Finally, using this expressions for $I_{51}$ and $I_{52}$ in equation \eqref{eq_outline_gdof_b5} we get the $5^{th}$ bound for the GDoF region.

{\it $6^{th}$ and $7^{th}$ bound}: Note that equations \eqref{eq_bound6} and \eqref{eq_bound7} involves terms whose approximations are already computed. Using those approximations we get the remaining 2 bounds of the GDoF region completing the proof.

\section{Proof of Corollary~\ref{claim:DoF-region}}
\label{App:proof-of-DoF-region}
We have to prove that the bounds in equations \eqref{eq:DoF-region-of-MIMO-IC-e}-\eqref{eq:DoF-region-of-MIMO-IC-g} are looser than the others. To analyze the 5-th bound, we start from equation \eqref{eq:DoF-region-of-MIMO-IC-e},
\begin{IEEEeqnarray*}{rl}
(d_1+d_2)\leq & (N_1\land M_2)+ ((N_1-M_2)^+\land (M_1-N_2)^+) + (N_2\land M_1)+ ((M_1-N_2)^+\land(M_2-N_1)^+);\\
=& \left\{\begin{array}{cc}
M_2+N_2+\{(N_1-M_2)\land (M_1-N_2)\},&\textrm{if }~N_1>M_2~\textrm{and}~M_1>N_2;\\
N_1+N_2,&\textrm{if }~N_1 \leq M_2~\textrm{and}~M_1>N_2;\\
M_1+M_2,&\textrm{if }~N_1>M_2~\textrm{and}~M_1\leq N_2;\\
M_1+N_1+\{(N_2-M_1)\land (M_2-N_1)\},&\textrm{if }~N_1\leq M_2~\textrm{and}~M_1 \leq N_2;
\end{array}\right.\\
=& \left\{\begin{array}{cc}
\min\{(M_1+M_2),(N_1+N_2)\},&\textrm{if }~N_1>M_2~\textrm{and}~M_1>N_2;\\
N_1+N_2,&\textrm{if }~N_1 \leq M_2~\textrm{and}~M_1>N_2;\\
M_1+M_2,&\textrm{if }~N_1>M_2~\textrm{and}~M_1\leq N_2;\\
\min\{(M_1+M_2),(N_1+N_2)\},&\textrm{if }~N_1\leq M_2~\textrm{and}~M_1 \leq N_2;
\end{array}\right.
\end{IEEEeqnarray*}
Clearly, this is looser than both the $ 3^{\rm rd}$ and the $4^{\rm th}$ bound. Consider next the $6^{\rm th}$ bound which, from equation \eqref{eq:DoF-region-of-MIMO-IC-f}, is given as

\begin{IEEEeqnarray*}{rl}
(2d_1+d_2)\leq & (N_1\land (M_1+M_2))+N_1\land (M_1-N_2)^++(N_2\land M_1)+ ((M_1-N_2)^+\land(M_2-N_1)^+);\\
=& \left\{\begin{array}{cc}
(N_1\land (M_1+M_2))+(N_1\land (M_1-N_2))+N_2 ,&\textrm{if }~M_1\geq N_2;\\
(N_1\land (M_1+M_2))+M_1+((N_2-M_1)\land (M_2-N_1)^+) ,&\textrm{if }~M_1 < N_2;
\end{array}\right.\\
=& \left\{\begin{array}{cc}
(N_1\land (M_1+M_2))+(N_1\land (M_1-N_2))+N_2 ,&\textrm{if }~M_1\geq N_2;\\
(N_1\land (M_1+M_2))+\min\{N_2,(M_1+M_2-N_1)\} ,&\textrm{if }~M_1 < N_2~\textrm{and}~M_2>N_1;\\
(N_1\land (M_1+M_2))+M_1 ,&\textrm{if }~M_1 < N_2~\textrm{and}~M_2\leq N_1;
\end{array}\right.\\
=& \left\{\begin{array}{cc}
(N_1\land (M_1+M_2))+\min\{(N_1+N_2),M_1\} ,&\textrm{if }~M_1\geq N_2;\\
\min\{(M_1+M_2),(N_1+N_2)\} ,&\textrm{if }~M_1 < N_2~\textrm{and}~M_2>N_1;\\
M_1+(N_1\land (M_1+M_2)) ,&\textrm{if }~M_1 < N_2~\textrm{and}~M_2\leq N_1;
\end{array}\right.\\
=& \left\{\begin{array}{cc}
(N_1\land (M_1+M_2))+\min\{(N_1+N_2),\max(N_2,M_1)\} ,&\textrm{if }~M_1\geq N_2;\\
\min\{M_1+\max(M_2,N_1),N_1+\max(N_2,M_1)\} ,&\textrm{if }~M_1 < N_2~\textrm{and}~M_2>N_1;\\
M_1+\min\{\max(N_1,M_2), (M_1+M_2)\} ,&\textrm{if }~M_1 < N_2~\textrm{and}~M_2\leq N_1;
\end{array}\right.
\end{IEEEeqnarray*}
It is clear that this bound is looser than the sum of the $1^{\rm st}$ and $3^{\rm rd}$ or the sum of the $1^{\rm st}$ and the $4^{\rm th}$ bounds. The proof of the fact that the same is true for the $7^{\rm th}$ bound is identical and is hence skipped.

\section{GDoF region of the 2-User MIMO MAC}
\label{App:GDoF:MIMO-MAC}

Both the set of lower and upper bounds to the capacity region of the 2-user MIMO IC contain terms that also appear in the capacity region of a 2-user MIMO MAC channel. Thus, as a by product we can obtain the GDoF region of the MIMO MAC channel.

Consider a MIMO MAC with two transmitters having $M_1$ and $M_2$ antennas, respectively, and with $N$ receive antennas at the common receiver. The input-output relation for this channel can be written as
\begin{eqnarray*}
Y=\sqrt{\rho}H X_{1}+ \sqrt{\rho^{\alpha}}G X_{2}+Z,
\end{eqnarray*}
where $X_i \in \mathbb{C}^{M_i\times 1}$ is the transmitted signal from user $i$, where $Y \in \mathbb{C}^{N\times 1}$ is the received signal and $H \in \mathbb{C}^{N\times M_1}$ and $G \in \mathbb{C}^{N\times M_2}$ are the channel matrices from users $1$ and $2$ to the receiver, respectively, both of which are assumed to have full rank, and $Z \sim \mathcal{CN}(\mathbf{0}, I_{N})$ is additive white Gaussian noise. Without loss of generality, we assume that the SNR of the second user is represented as $\rho^{\alpha}$.

Let $\mathcal{C}_{MAC}(H,G)$ denote the capacity region of the 2-user MIMO MAC defined above.
The GDoF region is defined as
\begin{eqnarray*}
\mathcal{D}_{MAC}=\left\{(d_1,d_2): ~d_i=\lim_{\rho \to \infty}\frac{R_i}{\log(\rho)}, i\in\{1,2\}~\textrm{and}~ (R_1,R_2)\in \mathcal{C}_{MAC}(H,G)\right\} .
\end{eqnarray*}

The result below gives the GDoF region of the 2-user MIMO MAC.
\begin{cor}
\label{cor_gdof_MIMO_MAC}
The GDoF region of the 2-user MIMO MAC defined above is given as
\begin{IEEEeqnarray*}{rl}
\Big\{(d_1,d_2): d_1 \leq & \min\{M_1,N\}; \\
d_2\leq & \min\{M_2,N\}\alpha;\\
(d_1+d_2)\leq & f\left(N, (\alpha,M_2),(1,M_1)\right) \Big\},
\end{IEEEeqnarray*}
where $f\left(.,.,.\right)$ is given by equation~\eqref{eq_def_f}.
\end{cor}

\begin{proof}
Following the analysis of the MIMO IC in \cite{Sanjay_Varanasi_Cap_MIMO_IC_const_gap}, it can be easily shown that an achievable rate region of the MIMO MAC
is given as
\begin{IEEEeqnarray*}{rl}
\mathcal{R}_A=\Big\{(R_1,R_2): R_1 \leq & \; \log \det \left(I_{N}+\rho HH^{\dagger}\right)-N\log(M_1); \\
 R_2  \leq & \; \log \det \left(I_{N}+\rho^{\alpha} GG^{\dagger}\right)-N \log(M_2);\\
R_1 +R_2 \leq & \; \log \det \left(I_{N}+\rho HH^{\dagger}+\rho^{\alpha} GG^{\dagger}\right)-N\log(\max\{M_1,M_2\}) \Big\},
\end{IEEEeqnarray*}
and an upper bound is given as
\begin{IEEEeqnarray*}{rl}
\mathcal{R}^U=\Big\{(R_1,R_2): R_1 \leq & \; \log \det \left(I_{N}+\rho HH^{\dagger}\right); \\
R_2\leq & \; \log \det \left(I_{N}+\rho^{\alpha} GG^{\dagger}\right);\\
R_1+R_2 \leq & \; \log \det \left(I_{N}+\rho HH^{\dagger}+\rho^{\alpha} GG^{\dagger}\right) \Big\}.
\end{IEEEeqnarray*}
Note that the two regions differ only by constant (independent of SNR) number of bits. The desired result now follows by replacing $\mathcal{C}_{MAC}$ in the definition of the GDoF region by $\mathcal{R}^U$ or $\mathcal{R}_A$, since a constant number of bits are insignificant in the GDoF analysis.
\end{proof}

\begin{rem}
The GDoF regions for the case when $N\geq (M_1+M_2)$ is depicted in Fig. \ref{fig_gdof_MACa} and the case when $\max\{M_1,M_2\}<N< (M_1+M_2)$ is depicted in Fig. \ref{fig_gdof_MACb}, where $A=(M_1,(N-M_1)\alpha)$, $B=((N-M_2)\alpha+M_1(1-\alpha),M_2\alpha)$, $A^{'}=(M_1,(N-M_1))$, $B^{'}=((N-M_2),M_2)$, $A^{''}=(M_1,(N-M_1)+M_2(\alpha-1))$ and $B^{''}=((N-M_2),M_2\alpha)$. Although, the GDoF analysis reveals the possibility of achieving a larger sum DoFs when one of the link's strength is exponentially larger that the other ($\alpha>1$), it is not as interesting as the MIMO IC since the GDoF region of the MAC can be achieved using independent Gaussian codes with scaled identity input covariances at each transmitter and joint decoding just as in a MAC with $\alpha=1$. In other words, this DoF-optimal scheme is also GDoF-optimal.

\begin{figure}[htp]
  \begin{center}
    \subfigure[$N\geq (M_1+M_2)$.]{\label{fig_gdof_MACa}\includegraphics[scale=.3]{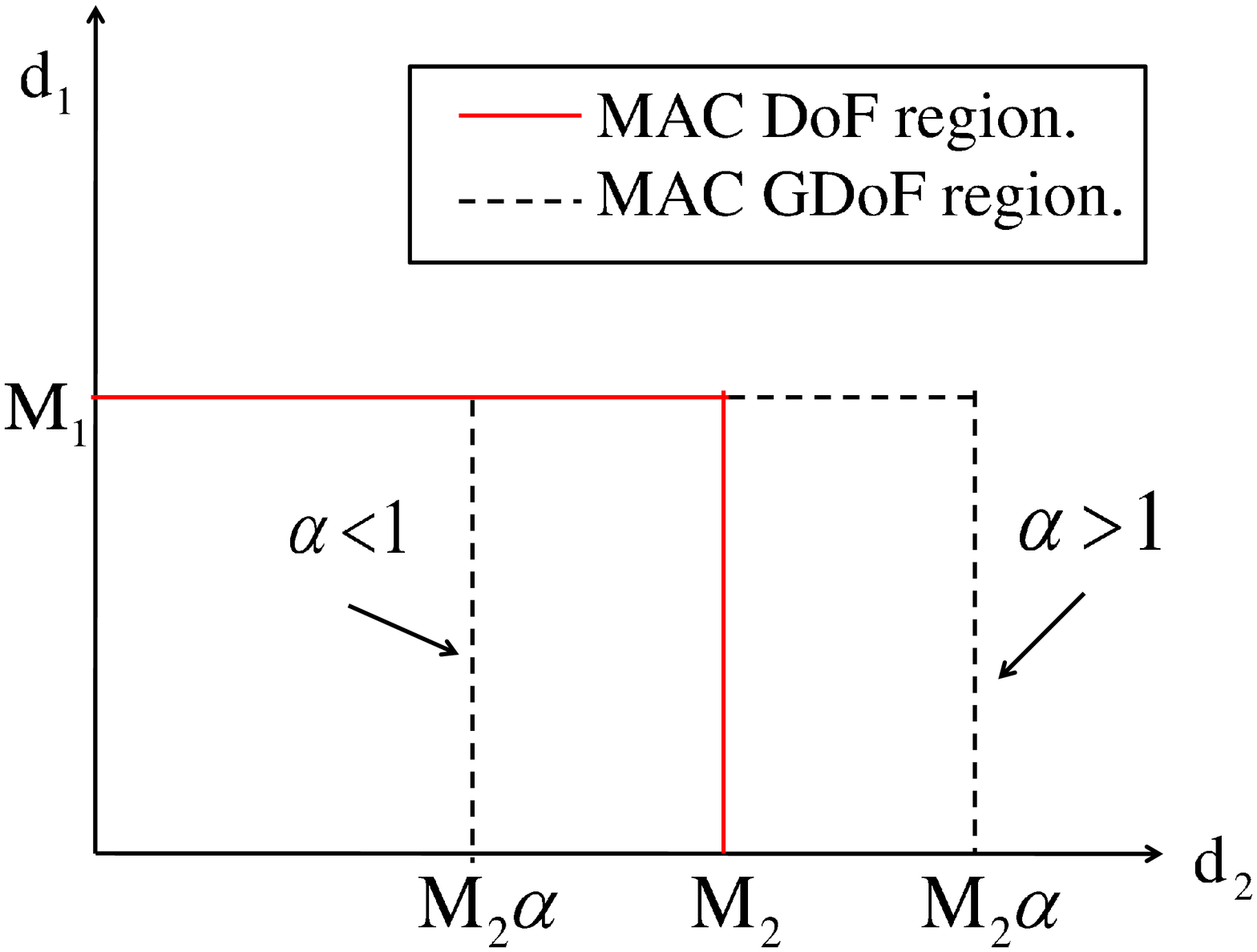}}
    \subfigure[$\max\{M_1,M_2\}<N< (M_1+M_2)$.]{\label{fig_gdof_MACb}\includegraphics[scale=0.3]{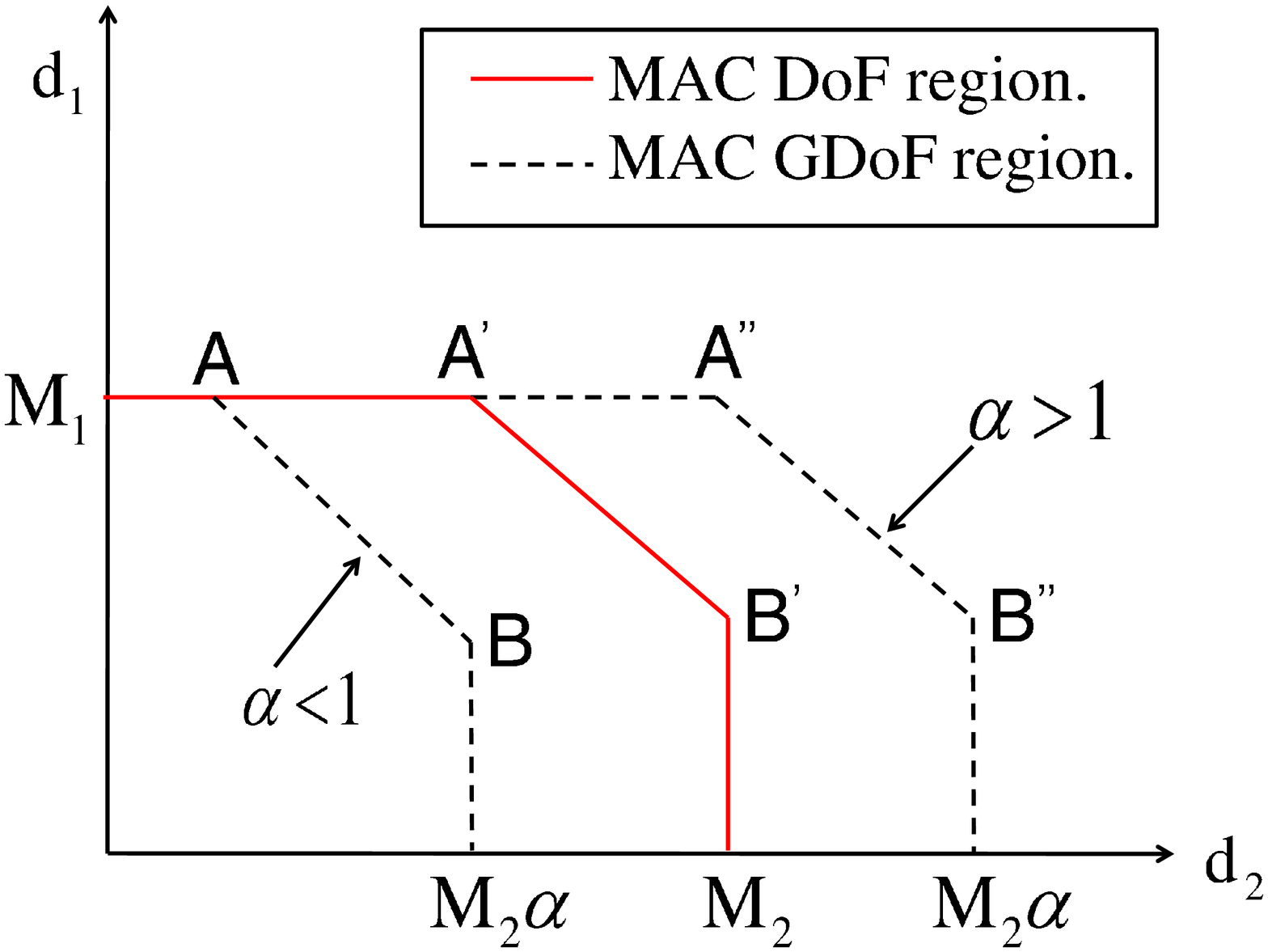}}
  \end{center}
\caption{The GDoF region of the MIMO MAC. }
\label{fig:GDoF-MIMO-MAC}
\end{figure}

\end{rem}

\section{Proof of Equivalent GDoF region}
\label{app:lem-dof-split}

In the HK coding scheme~\cite{Han_Kobayashi}, each user's message is divided into two parts called the private ($U_i$) and public ($W_i$) messages with rate $S_i$ and $T_i$, respectively. It was proved in \cite{Han_Kobayashi} that for any given probability distribution P(.) which factors as

\begin{IEEEeqnarray}{rl}
P(Q,W_1,U_1,W_2,U_2,X_1,X_2)=P(Q)P(U_1|Q)P(W_1|Q)&P(U_2|Q)P(W_2|Q)\nonumber\\
\label{eq:distribution-for-original-HK-scheme}
&P(X_1|U_1,U_2,Q)P(X_2|U_2,W_2,Q),
\end{IEEEeqnarray}
the rate region $\mathcal{R}^o_{\textrm{HK}}(P)=\mathcal{R}_{\textrm{HK}}^{(o,1)}(P)\cap \mathcal{R}_{\textrm{HK}}^{(o,2)}(P)$ is achievable where
\begin{IEEEeqnarray*}{rl}
\mathcal{R}_{\textrm{HK}}^{(o,i)}(P)=\Big\{(S_1,T_1,S_2,T_2):~S_i\leq & I(U_i;Y_i|W_i,W_j,Q);\\
T_i\leq & I(W_i;Y_i|U_i,W_j,Q);\\
T_j\leq & I(W_j;Y_i|W_i,U_i,Q);\\
(S_i+T_i)\leq & I(U_iW_i;Y_i|W_j,Q);\\
(S_i+T_j)\leq & I(U_iW_j;Y_i|W_i,Q);\\
(T_i+T_j)\leq & I(W_iW_j;Y_i|U_i,Q);\\
(S_i+T_i+T_j)\leq & I(U_iW_iW_j;Y_i|Q);\Big\}
\end{IEEEeqnarray*}
for $i\neq j\in\{1,2\}$. Let the 2-dimensional projection of the set $\mathcal{R}^o_{\textrm{HK}}(P)$ be denoted by $\Pi\left(\mathcal{R}^o_{\textrm{HK}}\right)$, which is defined as follows
\begin{equation*}
    \Pi\left(\mathcal{R}^o_{\textrm{HK}}(P)\right)=\{(0\leq R_1\leq (S_1+T_1),0\leq R_2\leq (S_2+T_2)): (S_1,T_1,S_2,T_2)\in \mathcal{R}_{\textrm{HK}}^o(P)\}.
\end{equation*}
Clearly, if $(R_1,R_2)\in \Pi\left(\mathcal{R}^o_{\textrm{HK}}(P)\right)$ then there exists a 4-tuple $(S_1,T_1,S_2,T_2)\in \mathcal{R}_{\textrm{HK}}^o(P)$ such that $(S_i+T_i)=R_i$ for $i=1,2$, and vice versa. This is true for any distribution satisfying \eqref{eq:distribution-for-original-HK-scheme}. Using the Fourier-Motzkin elimination method, a compact formula for the rate region $\Pi\left(\mathcal{R}^o_{\textrm{HK}}(P)\right)$ was recently derived in Lemma~1 of \cite{CMG}, which when evaluated for the input distributions specified in Section \ref{def_coding_scheme}, results in an achievable rate region containing the rate region given in Lemma~\ref{lem_achievable_region} (See Theorem~2 of \cite{Sanjay_Varanasi_Cap_MIMO_IC_const_gap}). Let us denote the rate region $\mathcal{R}^o_{\textrm{HK}}(P)$ by $\mathcal{R}_{\textrm{HK}}^G$, when $P$ is same as the distributions specified in Section \ref{def_coding_scheme}. Using the technique in the proof of Lemma~\ref{lem_achievable_region} (given in \cite{Sanjay_Varanasi_Cap_MIMO_IC_const_gap}) it then follows that $\mathcal{R}_{\textrm{HK}}^G=\mathcal{R}_{\textrm{HK}}^{G1}\cap \mathcal{R}_{\textrm{HK}}^{G2}$, where
\begin{IEEEeqnarray*}{rl}
\mathcal{R}_{\textrm{HK}}^{Gi}=\Big\{(S_i,T_i,T_{j}):~S_i\leq & \log \det \left(I_{N_i}+\rho_{ii} H_{ii}K_{iu}H_{ii}^{\dagger}+\rho_{ji}H_{ji}K_{ju}H_{ji}^\dagger\right)-\tau_{ji};\\
T_i\leq & \log \det \left(I_{N_i}+\rho_{ii} H_{ii}K_{iw}H_{ii}^{\dagger}+\rho_{ji}H_{ji}K_{ju}H_{ji}^\dagger\right)-\tau_{ji};\\
T_j\leq & \log \det \left(I_{N_i}+\frac{\rho_{ji}}{M_j} H_{ji}H_{ji}^{\dagger}\right)-\tau_{ji};\\
(S_i+T_i)\leq & \log \det \left(I_{N_i}+\frac{\rho_{ii}}{M_i} H_{ii}H_{ii}^{\dagger}+\rho_{ji}H_{ji}K_{ju}H_{ji}^\dagger\right)-\tau_{ji};\\
(S_i+T_j)\leq & \log \det \left(I_{N_i}+ \frac{\rho_{ji}}{M_j} H_{ji}H_{ji}^{\dagger}+  \rho_{ii} H_{ii} K_{iu} H_{ii}^{\dagger}\right)-\tau_{ji} \\
(T_i+T_j)\leq & \log \det \left(I_{N_i}+\frac{\rho_{ji}}{M_j} H_{ji}H_{ji}^{\dagger}+\rho_{ii} H_{ii}K_{iw}H_{ii}^{\dagger}\right)-\tau_{ji};\\
(S_i+T_i+T_j)\leq & \log \det \left(I_{N_i}+\frac{\rho_{ji}}{M_j} H_{ji}H_{ji}^{\dagger}+\frac{\rho_{ii}}{M_i} H_{ii}H_{ii}^{\dagger}\right)-\tau_{ji};\Big\}
\end{IEEEeqnarray*}
for $i\neq j\in\{1,2\}$ and $\tau_{ji}$'s for $1\leq i\neq j\leq 2$ are constants independent of $\rho$ or channel matrices. The GDoF region corresponding to the above achievable rate region can be defined as follows
\begin{IEEEeqnarray*}{rl}
\mathcal{G}_i(\bar{M},\bar{\alpha})=\Big\{(d_{1p},d_{1c}, d_{2p},d_{2c}):d_{ip}=\lim_{\rho_{ii} \to \infty}\frac{S_i}{\log(\rho_{ii})}, & d_{ic}=\lim_{\rho_{ii} \to \infty}\frac{T_i}{\log(\rho_{ii})}~\textrm{for}~i=1,2,\nonumber\\
&\textrm{and}~ (S_1,T_1,S_2,T_2)\in \mathcal{R}_{\textrm{HK}}^{Gi}\Big\}.
\end{IEEEeqnarray*}
Using this definition, and following a similar approach as in Theorem~\ref{thm_mainresult}, we get equation~\eqref{eq:GDoF-sub-region}. 

From the above analysis on one hand, we have the achievable rate region $\mathcal{R}_{\textrm{HK}}^G$ for the Gaussian IC, which is $\mathcal{R}^{o}_{\textrm{HK}}(P)$ evaluated for the distribution of Subsection \ref{def_coding_scheme}. on the other hand, we have $\mathcal{R}_a(\mathcal{H},\bar{\alpha})$, which is a {\it subset} of the rate region obtained when $\Pi\left(\mathcal{R}_{\textrm{HK}}^{o}(P)\right)$ is evaluated at the distribution in Subsection \ref{def_coding_scheme}. This two facts together imply that
\begin{equation}\label{}
    \mathcal{R}_a(\mathcal{H},\bar{\alpha})\subseteq \Pi\left(\mathcal{R}_{\textrm{HK}}^G\right),
\end{equation}
i.e., for any rate pair $(R_1,R_2)\in \mathcal{R}_a(\mathcal{H},\bar{\alpha})$ there exists a 4-tuple $(S_1,T_1,S_2,T_2)\in \mathcal{R}_{\textrm{\textrm{HK}}}^G$ such that $(S_i+T_i)=R_i$ for $i=1,2$. In other words, $\mathcal{R}_a(\mathcal{H},\bar{\alpha})$ is a subset of the 2-dimensional projection of the set $\mathcal{R}_{\textrm{HK}}^G$. Since $\mathcal{G}(\bar{M},\bar{\alpha})$ and $ \mathcal{D}_o(\bar{M},\bar{\alpha})$ are the high SNR scaled versions of the rate regions $\mathcal{R}_{\textrm{\textrm{HK}}}^G$ and $\mathcal{R}_a(\mathcal{H},\bar{\alpha})$, respectively the same is true for them. That is $ \mathcal{D}_o(\bar{M},\bar{\alpha})$ is a subset of the 2-dimensional projection of the set $\mathcal{G}(\bar{M},\bar{\alpha})$ or for every $(d_1,d_2)\in \mathcal{D}_o(\bar{M},\bar{\alpha})$, there exists a 4-tuple $(d_{1p},d_{1c},d_{2p},d_{2c})\in \mathcal{G}(\bar{M},\bar{\alpha})$ such that $(d_{ip}+d_{ic})=d_i$ for $i=1,2$.

\bibliographystyle{IEEEtran}
\bibliography{mybibliography}

\begin{thebibliography}{10}
\providecommand{\url}[1]{#1}
\csname url@samestyle\endcsname
\providecommand{\newblock}{\relax}
\providecommand{\bibinfo}[2]{#2}
\providecommand{\BIBentrySTDinterwordspacing}{\spaceskip=0pt\relax}
\providecommand{\BIBentryALTinterwordstretchfactor}{4}
\providecommand{\BIBentryALTinterwordspacing}{\spaceskip=\fontdimen2\font plus
\BIBentryALTinterwordstretchfactor\fontdimen3\font minus
  \fontdimen4\font\relax}
\providecommand{\BIBforeignlanguage}[2]{{%
\expandafter\ifx\csname l@#1\endcsname\relax
\typeout{** WARNING: IEEEtran.bst: No hyphenation pattern has been}%
\typeout{** loaded for the language `#1'. Using the pattern for}%
\typeout{** the default language instead.}%
\else
\language=\csname l@#1\endcsname
\fi
#2}}
\providecommand{\BIBdecl}{\relax}
\BIBdecl

\bibitem{Carleial75}
A.~B. Carleial, ``A case where interference does not reduce the capacity,''
  \emph{IEEE Trans. on Inform. Th.}, vol.~21, pp. 569--570, Sep, 1975.

\bibitem{Benzel}
R.~Benzel, ``The capacity region of a class of discrete additive degraded
  interference channels,'' \emph{IEEE Trans. on Inform. Th.}, vol.~25, pp.
  228--231, Mar, 1979.

\bibitem{Sato}
H.~Sato, ``The capacity of {G}aussian interference channel under strong
  interference,'' \emph{IEEE Trans. on Inform. Th.}, vol.~27, pp. 786--788,
  Nov, 1981.

\bibitem{Gamal_Costa}
A.~A.~E. Gamal and M.~H.~M. Costa, ``The capacity region of a class of
  deterministic interference channels,'' \emph{IEEE Trans. on Inform. Th.},
  vol.~28, pp. 343--346, Mar, 1982.

\bibitem{Motahari_Khandani}
A.~S. Motahari and A.~K. Khandani, ``Capacity bounds for the gaussian
  interference channel,'' \emph{IEEE Trans. on Inform. Th.}, vol.~55, pp.
  620--643, Feb, 2009.

\bibitem{Shang2009}
X.~Shang, G.~Kramer, and B.~Chen, ``A new outer bound and the
  noisy-interference sum-rate capacity for {G}aussian interference channels,''
  \emph{IEEE Trans. on Inform. Th.}, vol.~55, pp. 689--699, Feb, 2009.

\bibitem{Annapureddy_Veeravalli}
V.~S. Annapureddy and V.~V. Veeravalli, ``Gaussian interference networks: Sum
  capacity in the low-interference regime and new outer bounds on the capacity
  region,'' \emph{IEEE Trans. on Inform. Th.}, vol.~55, pp. 3032--3050, Jun,
  2009.

\bibitem{SCKP}
X.~Shang, B.~Chen, G.~Kramer, and H.~V. Poor, ``Capacity regions and sum-rate
  capacities of vector gaussian interference channels,'' \emph{IEEE Trans. on
  Inform. Th.}, vol.~56, pp. 5030--5044, Sep, 2010.

\bibitem{ETW1}
R.~Etkin, D.~Tse, and H.~Wang, ``Gaussian interference channel capacity to
  within one bit,'' \emph{IEEE Trans. on Inform. Th.}, vol.~54, pp. 5534--5562,
  Dec, 2008.

\bibitem{TT}
E.~Telatar and D.~N.~C. Tse, ``Bounds on the capacity region of a class of
  interference channels,'' in \emph{Proc. IEEE Int. Symp. on Inform. Th.}, Jun,
  2007.

\bibitem{Sanjay_Varanasi_Cap_MIMO_IC_const_gap}
S.~Karmakar and M.~K. Varanasi, ``Capacity of the {M}{I}{M}{O} interference
  channel to within a constant gap,'' Feb, 2011, preprint, available at
  \texttt{http://arxiv.org/abs/1102.0267}.

\bibitem{JFak}
S.~A. Jafar and M.~J. Fakhereddin, ``Degrees of freedom for the {M}{I}{M}{O}
  interference channel,'' \emph{IEEE Trans. on Inform. Th.}, vol.~53, pp.
  2637--2642, Jul, 2007.

\bibitem{Han_Kobayashi}
T.~S. Han and K.~Kobayashi, ``A new achievable region for the interference
  channel,'' \emph{IEEE Trans. on Inform. Th.}, vol.~27, pp. 49--60, Jan, 1981.

\bibitem{PBT}
P.~A. Parker, D.~W. Bliss, and V.~Tarokh, ``On the degrees-of-freedom of the
  {M}{I}{M}{O} interference channel,'' in \emph{Proc. of Inform. Sciences and
  Systems}, Mar, 2008, pp. 62--67.

\bibitem{Jafar_Vishwanath_GDOF}
S.~Jafar and S.~Vishwanath, ``Generalized degrees of freedom of the symmetric
  {Gaussian} {K} user interference channel,'' in \emph{arXiv:cs.IT/0804.4489},
  2008.

\bibitem{Bandemer-Cyclic}
B.~Bandemer, A.~E. Gamal, and G.~Vasquez-Vilar, ``On the sum capacity of the
  class of cyclically symmetric deterministic interference channels,'' in
  \emph{Proc. IEEE Int. Symp. on Inform. Th., Seoul, South Korea}, Jun, 2009,
  pp. 2622--2626.

\bibitem{Gou_Jafar}
T.~Gou and S.~A. Jafar, ``Capacity of a class of symmetric {S}{I}{M}{O}
  gaussian interference channels within ${O}(1)$,'' May, 2009, preprint,
  available at \texttt{http://arxiv.org/abs/0905.1745}.

\bibitem{Huang_Cadambe_Jafar}
C.~Huang, V.~R. Cadambe, and S.~A. Jafar, ``On the capacity and generalized
  degrees of freedom of the {X} channel,'' preprint, available at
  \texttt{http://arxiv.org/abs/0810.4741}.

\bibitem{Bresler-parekh-tse}
G.~Bresler, A.~Parekh, and D.~Tse, ``The approximate capacity of the
  many-to-one and one-to-many gaussian interference channels,'' \emph{IEEE
  Trans. on Inform. Th.}, vol.~56, pp. 4566--4592, Sep, 2010.

\bibitem{Cadambe-Jafar-IA}
V.~R. Cadambe and S.~A. Jafar, ``Interference alignment and degrees of freedom
  of the k-user interference channel,'' \emph{IEEE Trans. on Inform. Th.},
  vol.~54, pp. 3425--3441, Aug, 2008.

\bibitem{Tulino_Verdu}
A.~M. Tulino and S.~Verdu, \emph{Random Matrix Theory and Wireless
  Communications}.\hskip 1em plus 0.5em minus 0.4em\relax Now, 2004.

\bibitem{Huang_Jafar}
C.~Huang, S.~A. Jafar, S.~S. (Shitz), and S.~Vishwanath, ``On degrees of
  freedom region of {M}{I}{M}{O} networks without {C}{S}{I}{T},'' Sep. 2009,
  available Online: http://arxiv.org/pdf/0909.4017.

\bibitem{Vaze_Varanasi}
C.~S. Vaze and M.~K. Varanasi, ``The degrees of freedom regions of {M}{I}{M}{O}
  broadcast, interference, and cognitive radio channels with no {C}{S}{I}{T},''
  Sept. 2009, available Online: http://arxiv.org/abs/0909.5424.

\bibitem{ZG}
Y.~Zhu and D.~Guo, ``The degrees of freedom of mimo interference channels
  without state information at transmitters,'' Aug, 2010, preprint, available
  at \texttt{http://arxiv.org/abs/1008.5196}.

\bibitem{CMG}
H.~Chong, M.~Motani, H.~Garg, and H.~E. Gamal, ``On the {H}an-{K}obayashi
  region for the interference channel,'' \emph{IEEE Trans. on Inform. Th.},
  vol.~54, pp. 3188--3195, Jul, 2008.

\end{thebibliography}
\end{document}